\documentclass[a4paper,11pt]{article}

\usepackage{bbm}

\usepackage[a4paper,margin=1in]{geometry}
\usepackage[bookmarks,bookmarksopen,bookmarksdepth=2]{hyperref}
\usepackage{xspace}
\usepackage{comment}
\usepackage{amsmath}
\usepackage{amsfonts}
\usepackage{amssymb}
\usepackage{amsthm}
\usepackage{stmaryrd}
\usepackage{pdfsync}
\usepackage{thmtools,thm-restate}
\usepackage{paralist}
\usepackage{algorithm}
\usepackage[noend]{algpseudocode}
\usepackage[utf8]{inputenc}
\usepackage[numbers]{natbib}

\declaretheorem{proposition}
\declaretheorem[sibling=proposition]{lemma}
\declaretheorem[sibling=proposition]{theorem}
\declaretheorem[sibling=proposition]{corollary}
\declaretheorem[sibling=proposition,style=remark]{remark}

\newcommand{\N}{\mathbb{N}}
\newcommand{\Z}{\mathbb{Z}}

\newcommand{\inc}[1]{\textbf{inc$(#1)$}}
\newcommand{\dec}[1]{\textbf{dec$(#1)$}}
\newcommand{\test}[1]{\textbf{test$(#1)$}}

\makeatletter
\def\BState{\State\hskip-\ALG@thistlm}
\makeatother

\newcommand{\Loopatmost}[2]{\Loop\ \textbf{at most $#1+#2$ times}}

\algrenewcommand{\algorithmiccomment}[1]{\qquad$\rightarrow$ #1}

\newcommand{\vr}[1]{\mathsf{#1}}

\newcommand{\size}[1]{\operatorname{size}(#1)}

\newcommand{\setN}{\mathbb{N}}

\title{The Reachability Problem for Petri Nets is Not Primitive Recursive%
\thanks{This research has been supported by ANR programme BraVAS (ANR-17-CE40-0028).}}

\author{%
  J\'er\^ome Leroux     \\ CNRS \& University of Bordeaux \\ \texttt{jerome.leroux@labri.fr}}

\date{}

\begin{document}


\pagenumbering{roman}

\maketitle

\newcommand{\cmd}{\textbf{cmd}\xspace}
\newcommand{\crochet}[1]{\left\llbracket#1\right\rrbracket}

\newcommand{\reduce}{\operatorname{evalF}}
\newcommand{\redmax}{\operatorname{evalF}^{\max{}}}
\newcommand{\unit}{1}
\newcommand{\zero}{0}

\begin{abstract}
  We provide an Ackermannian complexity lower bound for the reachability problem for checking programs, a model equivalent to Petri nets. Moreover in fixed dimension $2d+4$, we show that the problem is $\mathbb{F}_d$-hard. As a direct corollary, the reachability problem in dimension 10 is not elementary.
  \end{abstract}

\pagenumbering{arabic}

\section{Introduction}
Checking programs~\cite{leroux:FOCS22}, or equivalently vector addition systems with states~\citep{Hopcroft79}, or vector addition systems~\citep{DBLP:journals/jcss/KarpM69}, or Petri nets are one of the most popular formal methods for the representation and the analysis of parallel processes \cite{survey-esparza}.

\medskip

Those equivalent models are acting on counters ranging over the natural numbers thanks to increment, decrement, and test commands (only at the end of an execution). The central algorithmic problem for checking programs is reachability: given a checking program, decide whether there exists an execution from an initial configuration to a final one. Many important computational problems in logic and complexity reduce or are even equivalent to this problem~\cite{DBLP:journals/siglog/Schmitz16,hack75}. After an incomplete proof by Sacerdote and Tenney~\cite{sacerdote77}, decidability of the
problem was established by Mayr~\cite{Mayr81,Mayr84}, whose proof was
then simplified by Kosaraju~\cite{Kosaraju82}.  Building on the
further refinements made by Lambert in the 1990s~\cite{Lambert92},
in 2015, a first complexity upper bound of the reachability problem was provided~\cite{DBLP:conf/lics/LerouxS15} more than thirty years after
the presentation of the algorithm introduced by Mayr~\cite{Kosaraju82,
  Lambert92}. The upper bound given in that paper is ``cubic
Ackermannian'', i.e. in $\mathbb{F}_{\omega^3}$~(see~\cite{Schmitz16toct}).
This complexity bound is obtained by analyzing the 
Mayr algorithm.
With a refined algorithm and a new ranking function for proving
termination, an Ackermannian complexity upper bound was obtained~\cite{DBLP:conf/lics/LerouxS19}
This
means that the reachability problem can be solved in time bounded by $A(p(n))$
where $p$ is a primitive recursive function and where $A$ is the Ackermann
function. This paper
also showed that the reachability problem in fixed
dimension $d$ (the dimension is the number of used counters)
is primitive recursive by bounding the length of executions by $O(F_{d+4}(n))$ where $F_{d+4}$ is a primitive recursive function of the Grzegorczyk hierarchy~(see Section~\ref{sec:fast} for the definition of those functions), and $n$ is the size of the input.

\medskip

Concerning the complexity lower bound, in 1976, the reachability problem was proved to be ExpSpace-hard~\cite{DBLP:conf/stoc/CardozaLM76}. This bound used to be the best one for more than forty years until 2019 when it was improved to a Tower complexity lower bound~\cite{DBLP:conf/stoc/CzerwinskiLLLM19}, i.e. a non elementary complexity.

\medskip

\textbf{Contributions}. In this paper we provide an Ackermannian complexity lower bound for the reachability problem for checking programs closing the gap with the Ackermannian complexity upper bound solving a 45 years old open problem. Moreover, in fixed dimension $2d+4$ with $d\geq 3$, we prove that the reachability problem is hard for the complexity class $\mathbb{F}_d$ introduced in \cite{Schmitz16toct} associated with the function $F_d$. As a direct corollary, we derive that the reachability problem is not elementary in dimension 10.

\medskip

This paper provides the last piece leading to the exact complexity of the reachability problem for checking programs. As previously mentioned, it follows a long series of results. This piece of work is not the most difficult one but it is an important one since it closes a long standing open problem. Technically, the most difficult piece is the notion of $K$-amplifiers introduced in the Tower complexity lower bound paper~\cite{DBLP:conf/stoc/CzerwinskiLLLM19} that provides a way to postpone at the end of an execution the tests commands of a Minsky machine (a machine like a VASS but with commands that can test to zero some counters) with counters bounded by $K$.

\medskip

In this paper, we provide several gadgets for proving an $\mathbb{F}_d$ complexity lower bound for the reachability problem in dimension $2d+4$. Recently and independently, other gadgets for implementing such a bound have been introduced in order to obtain the same complexity upper bound by Wojciech Czerwiński, and Łukasz Orlikowski~\cite{DBLP:journals/corr/abs-2104-13866} in dimension $6d$, by myself~\cite{leroux:FOCS22} in dimension $4d+5$, and by Sławomir Lasota~\cite{lasota2021improved} in dimension $3d+2$. We do not know if the lower bound provided in this paper, i.e. $\mathbb{F}_d$ in dimension $2d+4$ is optimal since the complexity upper bound is $\mathbb{F}_{d+4}$ in dimension $d$. It follows that the parameterized complexity (i.e. in fixed dimension) of the reachability problem is still open. 

\medskip

We think that several different solutions to the complexity lower bound for the reachability problem is useful not only for the confidence in the claimed result but also for future work. In fact, the reachability problem for many extensions is open for almost all natural extensions except for vector addition systems with hierarchical zero tests~\cite{DBLP:conf/mfcs/Bonnet11,REINHARDT2008239}. Moreover, the best known complexity lower bounds for those models only come from checking programs. Finding gadgets that take benefits from the extra power given by the considered extensions is an open problem.

\medskip

\textbf{Outline}.
In Section~\ref{sec:fast}, we recall some properties satisfied by the fast growing functions $F_d$ and introduce a way to compute $F_d$ be iterating a reduction rule $\reduce_d$. In Section~\ref{sec:program} we introduce the model of general programs, and the subclasses of test-free models that correspond to programs that cannot test counters to zero, and the checking programs that can only test counters at the end. Whereas the reachability problem for test-free program is equivalent to the so-called \emph{coverability problem} for Petri nets (see \cite{DBLP:conf/stoc/CardozaLM76,DBLP:journals/tcs/Rackoff78} for complexity results), the reachability problem  for checking program is equivalent to the so-called reachability problem for Petri nets. In Section~\ref{sec:test} we provide tools to postpone at the end of an execution test commands of a general program. Those tools are used in Section~\ref{sec:bounded} in order to simulate the bounded semantics of general programs thanks to the so-called preamplifiers. In Section~\ref{sec:loop} we provide tools for iterating a test-free program a fixed number of times that depends on the valuation of some counters. Those tools are used in Section~\ref{sec:reduce} to implements $\reduce_d$ thanks to a test-free program. By iterating this test-free program, we provide a way in Section~\ref{sec:ack} to implement an Ackermannian preamplifier. Finally, in Section~\ref{sec:complexity} we collect intermediate results to provide complexity lower bounds for the reachability problem for checking programs.

\section{Fast Growing Functions}\label{sec:fast}
We introduce the sequence $(F_d)_{d\in\setN}$ of functions $F_d:\setN\rightarrow\setN$ defined by $F_0(n)=n+1$, and defined by induction on $d\geq 1$ by $F_d(n)=F^{n+1}_{d-1}(n)$. Rather than take the original definition of \emph{Ackermann function}, we let $A(n) = F_\omega(n)$ defined by $F_\omega(n)=F_{n+1}(n)$, which behaves like the classical function for our complexity-theoretical purpose. 
We introduce the function $F_v:\setN\rightarrow\setN$ with $v\in\setN^d$, defined as follows for every $n\in\setN$:
$$F_v(n)=F_d^{v[d]}\circ\cdots\circ F_1^{v[1]}(n)$$

\begin{remark}
  We have $F_1(n)=2n+1$ and $F_2(n)=2^{n+1}(n+1)-1$ for every $n\in\setN$. The function $F_3$ behaves like a tower of $n$ exponential and it is not elementary. Each function $F_d$ is primitive recursive and $F_\omega$ is not primitive recursive.
\end{remark}

\newcommand{\val}{\operatorname{val}}
\newcommand{\lex}{\leq_{\textrm{lex}}}
\newcommand{\slex}{<_{\textrm{lex}}}

Let us denote by $\zero_d$ the zero vector of $\setN^d$, and by $\unit_{d,i}$ the $i$th unit vector of $\setN^d$ defined for every $1\leq j\leq d$ by $\unit_{d,i}[j]=1$ if $j=i$ and $\unit_{d,i}[j]=0$ otherwise. The lexicographic (strict) order $<_{\textrm{lex}}$ over $\setN^d$ is defined by $w\slex v$ if $w\not= v$ and the maximal $p\in\{1,\ldots,d\}$ such that $w[p]\not=v[p]$ satisfies $w[p]<v[p]$.

\medskip

Let $\reduce_d:\setN^d\times \setN\rightarrow\setN^d\times\setN$ be the function partially defined over the pairs $(v,n)\in\setN^d\times\setN$ such that $v\not=0_d$ as follows where $p$ is the minimal index in $\{1,\ldots,d\}$ such that $v[p]>0$:
$$\reduce_d(v,n) =
\begin{cases}
  (v-\unit_{d,p},2n+1) & \text{ if $p=1$}\\
  (v-\unit_{d,p}+(n+1)\unit_{d,p-1},n) & \text{ if $p>1$}
\end{cases}
$$
Since $F_p(n)=F_{p-1}^{n+1}(n)$, one can easily prove (see~\cite{DBLP:conf/mfcs/Schnoebelen10} for more details) that if $v\not=\zero_d$, then the pair $(w,m)$ defined as $\reduce_d(v,n)$ satisfies $F_w(m)=F_v(n)$. It follows that the function that maps $(v,n)$ onto $F_v(n)$ is an invariant of $\reduce_d$.

\medskip

Moreover, as $w\slex v$, it follows that we can iterate the function $\reduce_d$ on a pair $(v,n)$ only a finite number of times (we use the well-foundedness of the lexicographic order). Let us introduce the function $\redmax_d:\setN^d\times \setN\rightarrow\setN^d\times \setN$ defined by $\redmax_d(v,n)=\reduce_d^k(v,n)$ where $k$ is the maximal number of times (it can be zero) we can apply the function $\reduce_d$ on $(v,n)$. Notice that $\redmax_d(v,n)$ is a pair of the form $(w,m)$ for some pair $(w,m)\in\setN^d\times\setN$. Moreover, by maximality of $k$, it follows that $w=\zero_d$. Since the function $(v,n)\mapsto F_v(n)$ is an invariant of $\reduce_d$, we deduce that $F_{\zero_d}(m)=F_v(n)$. It follows that the following equality holds~(see~\cite{DBLP:conf/mfcs/Schnoebelen10} for more details):
$$\redmax_d(v,n)=(\zero_d,F_v(n))$$

As a direct corollary, we deduce the following lemma.
\begin{lemma}\label{lem:fast}
  We have $\redmax_d((n+1)\unit_{d,d},n)=(\zero_d,F_{d+1}(n))$ for every $n\geq 0$ and $d\geq 1$.
\end{lemma}

\newcommand{\norm}[1]{|#1|}

\medskip

The following lemma will be useful in the sequel. In that lemma, $\norm{v}$ is $v[1]+\cdots+v[d]$ for any vector $v\in\setN^d$.
\begin{lemma}\label{lem:tech}
  We have $F_v(n)\geq 2^{\norm{v}}n+\norm{v}$ for every $v\in\setN^d$ and for every $n\in\setN$.
\end{lemma}
\begin{proof}
  Let $v\in\setN^d$ and $n\in\setN$. Assume that for every $w\in\setN^d$ such that $w\slex v$ we have $F_w(m)\geq 2^{\norm{w}}m+\norm{w}$ for every $m\in\setN$. And let us prove that in that case $F_v(n)\geq 2n^{\norm{v}}+\norm{v}$ for every $n\in\setN$. In fact, since the lexicographic order  is well-founded, by induction the proof of the lemma reduces to that statement.
  
  Notice that if $v=\zero_d$ then $F_v(n)=n$ and we are done. So, we can assume that $v\not=\zero_d$. In that case, let us introduce $(w,m)=\reduce_d(v,n)$. Since $w\slex v$, it follows that $F_w(m)\geq 2^{\norm{w}}m+\norm{w}$. Since $F_{w}(m)=F_v(n)$, we get $F_v(n)\geq 2^{\norm{w}}m+\norm{w}$. Let $p\in\{1,\ldots,d\}$ be the minimal index such that $v[p]>0$. Observe that if $p=1$ then $\norm{w}=\norm{v}-1$ and $m=2n+1$. In particular $2^{\norm{w}}m+\norm{w}\geq 2^{\norm{v}}n+\norm{v}$. If $p>1$ then $\norm{w}=\norm{v}+n$ and $m=n$. Hence $2^{\norm{w}}m+\norm{w}\geq 2^{\norm{v}}n+\norm{v}$. We have proved that $F_v(n)\geq 2^{\norm{v}}n+\norm{v}$ in any case.
\end{proof}

\section{General Programs}\label{sec:program}
We introduce in the section the formal model of programs acting on counters ranging over the natural numbers.

\medskip

Formally, an implicit infinite countable set of elements called \emph{counters} is given. A \emph{configuration} is a function $\rho$ that maps every counter on a natural number in such a way that the set of counters $\vr{c}$ such that $\rho(\vr{c})\not=0$ is finite. We denote by $0$ the configuration $\rho$ such that $\rho(\vr{c})=0$ for every counter $\vr{c}$. The \emph{increment, decrement and test commands} of a counter $\vr{c}$ are respectively denoted as \inc{\vr{c}}, \dec{\vr{c}}, and $\test{\vr{c}}$. We associate with such a command $\cmd$, the binary relation $\xrightarrow{\cmd}$ over the configurations defined by $\alpha\xrightarrow{\cmd}\beta$ if $\alpha(\vr{x})=\beta(\vr{x})$ for every counter $\vr{x}\not=\vr{c}$, and satisfying additionally:
$$\begin{cases}
  \beta(\vr{c})=\alpha(\vr{c})+1 & \text{ if $\cmd$ is \inc{\vr{c}}}\\
  \beta(\vr{c})=\alpha(\vr{c})-1 & \text{ if $\cmd$ is \dec{\vr{c}}}\\
  \beta(\vr{c})=0=\alpha(\vr{c}) & \text{ if $\cmd$ is \test{\vr{c}}}\\
\end{cases}$$

\medskip

A \emph{(general) program} $M$ is defined as an increment/decrement/test command, or inductively as a loop program $\textbf{loop}~M_0$, a series composition $M_1;M_2$, or a non-deterministic choice $M_1\textbf{ or }M_2$ where $M_0,M_1,M_2$ are programs. The \emph{size} of a program $M$ is the number $\size{M}$ defined inductively as $1$ if $M$ is a command, $1+\size{M_0}$ if $M=\textbf{loop}~M_0$, and $1+\size{M_1}+\size{M_2}$ if $M=M_1;M_2$ or $M=M_1\textbf{ or }M_2$. The \emph{dimension} of $M$ is the cardinal of the set of counters used by $M$, and formally defined as expected. We associate with every program $M$ the binary relation $\xrightarrow{M}$ over the configurations defined as follows:
$$
\xrightarrow{M}~~=~~\begin{cases}
  \xrightarrow{\textbf{cmd}} & \text{ if }M=\textbf{cmd}\text{ is a command}\\
  (\xrightarrow{M_0})^* & \text{ if }M=\textbf{loop}~M_0\\
  \xrightarrow{M_1};\xrightarrow{M_2} & \text{ if }M=M_1;M_2\\
  \xrightarrow{M_1}\cup\xrightarrow{M_2} & \text{ if }M=M_1\textbf{ or }M_2\\
\end{cases}
$$
Where $(\rightarrow_0)^*$ is the reflexive and transitive closure of  $\rightarrow_0$, and $\rightarrow_1;\rightarrow_2$ is defined by $\alpha\rightarrow_1;\rightarrow_2\beta$ if there exists a configuration $\rho$ such that $\alpha\rightarrow_1\rho$ and $\rho\rightarrow_2\beta$.
Series compositions and non-deterministic choices are clearly associative with respect to the relation $\rightarrow$. In particular given a sequence $M_1,\ldots,M_n$ of Minsky programs, the relations $\xrightarrow{M_1;\ldots;M_n}$ and $\xrightarrow{M_1\textbf{ or }\ldots\textbf{ or }M_n}$ are well defined. We denote by $M^{(n)}$ the series composition of $M$ by itself $n$ times. We also denote by $\inc{\vr{c}_1,\ldots,\vr{c}_n}$ the program $\inc{\vr{c}_1};\ldots;\inc{\vr{c}_n}$, by $\dec{\vr{c}_1,\ldots,\vr{c}_n}$ the program $\dec{\vr{c}_1};\ldots;\dec{\vr{c}_n}$, and similarly $\test{\vr{c}_1,\ldots,\vr{c}_n}$ the program $\test{\vr{c}_1};\ldots;\test{\vr{c}_n}$.

\medskip

We say that a program is \emph{test-free} if it does not use any test command. A \emph{checking program} is a program of the form $M;\test{\vr{c}_1,\ldots,\vr{c}_n}$ where $M$ is a test-free program and $\vr{c}_1,\ldots,\vr{c}_n$ are some counters. 

\medskip

The \emph{reachability problem} for general programs asks, given a general program $M$, whether there exists a configuration $\beta$ such that $0\xrightarrow{M}\beta$. This problem is undecidable~\cite{minsky1967computation} even in dimension $2$ but provides a way to define complexity classes beyond Elementary as recalled in Section~\ref{sec:complexity}. The reachability problem for test-free programs is equivalent (i.e. inter-reducible) to the so-called coverability problem for Petri nets (see~\cite{DBLP:conf/stoc/CardozaLM76,DBLP:journals/tcs/Rackoff78} for the complexity of the problem), and the reachability problem for checking programs is equivalent (i.e. inter-reducible) to the so-called reachability problem for Petri nets.

\medskip

In this paper, we use standard notions from model theory by considering counters as variables, and configurations as valuations of the variables. It means that a configuration is used to replace in an expression $e$ over the counters, each occurrence of a counter $\vr{c}$ by $\rho(\vr{c})$. We denote by $\rho(e)$ the expression we obtain this way. When $\phi$ is a constraint over the counters, we also denote by $\rho(\phi)$ the constraint we obtain by considering $\phi$ as an expression over the counters. We say that $\rho$ satisfies $\phi$ if the constraint $\rho(\phi)$ is true. For instance we say that $\rho$ satisfies $\vr{y}=2^{\vr{c}}\vr{x}$ where $\vr{y},\vr{c}$ and $\vr{x}$ are counters if $\rho(\vr{y})=2^{\rho(\vr{c})}\rho(\vr{x})$.

\newcommand{\progline}[2]{\progcore{#1}{#2}{[1]}}
\newcommand{\prog}[2]{\progcore{#1}{#2}{}}

\newcommand{\progcore}[3]{
  \begin{quote}
    \begin{tabular}{@{}l|@{}l@{}}$#1~~=~~$ &
      \begin{minipage}{13cm}
        \begin{algorithmic}#3
          #2
        \end{algorithmic}
      \end{minipage}
    \end{tabular}
  \end{quote}
}

\section{Simulating Test Commands}\label{sec:test}
We provide in this section the central idea for simulating test commands. Assume that $\vr{B}$ denotes a finite set of counters. The simulation of a test command $\test{\vr{c}}$ for some counter $\vr{c}\in\vr{B}$ is using two distinct auxiliary counters $\vr{x}$ and $\vr{y}$ not in $\vr{B}$. During the simulation, we consider configurations $\rho$ satisfying $\vr{y}\geq K\vr{x}$ where $K=\rho(\sum_{\vr{b}\in\vr{B}}\vr{b})$. Intuitively, when the configuration satisfies $\vr{y}=K\vr{x}$ it means that all the previous simulation of test commands were correct and when $\vr{y}>K\vr{x}$ is means that at least one simulation was incorrect. Those simulations are obtained by introducing the following test-free program. This program is using any implicit enumeration $\vr{b}_0,\ldots,\vr{b}_d$ of the counters in $\vr{B}$ satisfying $\vr{b}_0=\vr{c}$. Such an enumeration naturally depends on the counter $\vr{c}$.

\prog{\textbf{simtest}_{\vr{y},\vr{B},\vr{x}}(\vr{c})}{
  \Loop
  \State \dec{\vr{b}_{1}};\inc{\vr{b}_0};\dec{\vr{y}}
  \EndLoop
  \State $\vdots$
  \Loop
  \State \dec{\vr{b}_{d}};\inc{\vr{b}_{d-1}};\dec{\vr{y}}
  \EndLoop
  \Loop
  \State \dec{\vr{b}_{d-1}};\inc{\vr{b}_{d}};\dec{\vr{y}}
  \EndLoop
  \State $\vdots$
  \Loop
  \State \dec{\vr{b}_{0}};\inc{\vr{b}_{1}};\dec{\vr{y}}
  \EndLoop
  \State $\dec{\vr{x}}^{(2)}$
}

Intuitively this program is transferring the content of the counter $\vr{b}_i$ into $\vr{b}_{i-1}$ thanks to the $i$th loop with $i\in\{1,\ldots,d\}$. Then the program is transferring back the contents of the counters thanks to the last $d$ loops. During each step of the transfer, the counter $\vr{y}$ is decremented. Since $K=\sum_{\vr{b}\in\vr{B}}\vr{b}$ is an invariant, the counter $\vr{y}$ can be decremented by at most $K$ in total with the first $d$ loops, and by $K$ as well in total with the last $d$ loops. It follows that $\vr{y}$ can be decremented by at most $2K$. Moreover, the only way to decrement $\vr{y}$ by exactly $2K$ is when initially $\vr{c}=0$ and when all loops are executed the maximal number of times, meaning that the transfers are total, and in particular that the initial and the final configurations coincide on the counters in $\vr{B}$. Since $\vr{x}$ is decremented by $2$, it follows that if $\vr{y}\geq K\vr{x}$ before the execution of the program, then the same constraint holds at the end. Moreover, if additionally $\vr{y}=K\vr{x}$ at the end, then necessarily the same constraint holds at the beginning and it means that $\vr{y}$ has been decremented by $2K$, and in particular initially $\vr{c}=0$ and all the transfers were performed maximally. The following two lemmas formally proved those intuitions.

\begin{lemma}\label{lem:exist}
  Let $\alpha,\beta$ be two configurations such that $\alpha\xrightarrow{\test{\vr{c}};\dec{\vr{y}}^{(2K)}\dec{\vr{x}}^{(2)}}\beta$ with $\vr{c}\in\vr{B}$ and such that $K=\alpha(\sum_{\vr{b}\in\vr{B}}\vr{b})$. Then we have:
  $$\alpha\xrightarrow{\textbf{simtest}_{\vr{y},\vr{B},\vr{y}}(\vr{c})}\beta$$
\end{lemma}
\begin{proof}
  We are going to prove that there exists an execution in such a way the loops are executed $m_1,\ldots,m_d,n_d,\ldots,n_1$ times (in that order), where $m_i=n_i=\alpha(\vr{b}_i)$ for every $1\leq i\leq d$. Notice that the decrement commands $\dec{\vr{y}}$ cannot prevent such an execution since $\alpha(\vr{y})\geq 2K=\sum_{i=1}^d(n_i+m_i)$. Moreover, the decrement commands $\dec{\vr{x}}^{(2)}$ are also executable at the end of the program since $\alpha(\vr{x})\geq 2$. So, for proving the existence of an execution of the program with the loops executed the right number of times, we can forget the decrement commands on the counters $\vr{y}$ and $\vr{x}$.

  For the first $d$ loops, such an execution is possible by observing that the $i$th loop with $i\in\{1,\ldots,d\}$ only decrements the counter $\vr{b}_i$ (recall that we forget the counter $\vr{y}$ for the proof) that is untouched by the previous loops. After executing those $d$ first loops, we get a configuration $\rho$ satisfying for every counter $\vr{c}$ the following equalities:
  $$\rho(\vr{c})=
  \begin{cases}
    \alpha(\vr{b}_{i+1}) & \text{ if $\vr{c}=\vr{b}_i$ with $0\leq i<d$}\\
    0 & \text{ if $\vr{c}=\vr{b}_d$}\\
    \alpha(\vr{y})-K & \text{ if }\vr{c}=\vr{y}\\
    \alpha(\vr{c}) & \text{ otherwise}
  \end{cases}$$
  Starting form $\rho$, the last $d$ loops are symmetrical to the first $d$ loops. In particular, the same argument as previously mentioned shows that there exists an execution from $\rho$ that executes the last $d$ loops the right number of times. Now, just observe that with such an execution, we have proved the lemma.
\end{proof}
 
\begin{lemma}\label{lem:forall}
  Let $\alpha,\beta$ be two configurations such that $\alpha\xrightarrow{\textbf{simtest}_{\vr{y},\vr{B},\vr{x}}(\vr{c})}\beta$ and let $K=\alpha(\sum_{\vr{b}\in\vr{B}}\vr{b})$. Then $\beta(\sum_{\vr{b}\in\vr{B}}\vr{b})=K$ and if $\alpha$ satisfies $\vr{y}\geq K\vr{x}$ then $\beta$ satisfies the same constraint. Moreover if additionally $\beta$ satisfies $\vr{y}=K\vr{x}$ then $\alpha$ satisfies the same constraint and:
  $$\alpha\xrightarrow{\test{\vr{c}};\dec{\vr{y}}^{(2K)}\dec{\vr{x}}^{(2)}}\beta$$
\end{lemma}
\begin{proof}
  Let us denote by $m_1,\ldots,m_d,n_d,\ldots,n_1$ the number of times loops of the program are executed (in that order). Observe that the $i$th loop with $i\in\{1,\ldots,d\}$ decrements the counter $\vr{b}_i$, and that counter is untouched by the previous loops. It follows that there exists a sequence $r_1,\ldots,r_d$ of natural numbers such that $m_i=\alpha(\vr{b}_i)-r_i$ for every $i\in\{1,\ldots, d\}$. We derive the following equality by observing that $\sum_{i=0}^d\alpha(\vr{b}_i)=K$:
  $$\sum_{i=1}^d m_i=K-\alpha(\vr{b}_0)-\sum_{i=1}^d r_i$$
  
  Let us denote by $\rho$ the configuration we obtain just after executing those $d$ first loops.   Since $\sum_{i=0}^{d}\vr{b}_i$ is an invariant of every line of the program, we get $\beta(\sum_{i=0}^d\vr{b}_i)=K$ and $\rho(\sum_{i=0}^d\vr{b}_i)=K$. Now, observe that the argument used in the previous paragraph holds for the last $d$ loops. It follows that there exists a sequence $s_1,\ldots,s_d$ of natural numbers such that $n_i=\rho(\vr{b}_{i-1})-s_i$ for every $1\leq i\leq d$, and we derive the following equality:
  $$\sum_{i=1}^d n_i=K-\rho(\vr{b}_d)-\sum_{i=1}^d s_i$$

  From $\beta(\vr{y})=\alpha(\vr{y})-\sum_{i=1}^d(m_i+n_i)$ and $\beta(\vr{x})=\alpha(\vr{x})-2$, we deduce the following equalities:
  \begin{align*}
    \beta(\vr{y}-K\vr{x})
    &=
    \alpha(\vr{y}-K\vr{x}) + 2K-\sum_{i=1}^n (m_i+n_i)\\
    & = \alpha(\vr{y}-K\vr{x})+\alpha(\vr{b}_0)+\rho(\vr{b}_d)+\sum_{i=1}^d(r_i+s_i)
  \end{align*}
  
  Now, assume that $\alpha$ satisfies $\vr{y}\geq K\vr{x}$. From the previous equality, we deduce that $\beta(\vr{y}-K\vr{x})$ is a sum of natural numbers. In particular $\beta$ satisfies $\vr{y}\geq K\vr{x}$. If additionally $\beta$ satisfies $\vr{y}=K\vr{x}$ the previous equality shows that $\alpha(\vr{y}-K\vr{x})= 0$, $\alpha(\vr{b}_0)= 0$, $\rho(\vr{b}_d)=0$, and $r_i=s_i=0$ for every $1\leq i\leq d$. From $r_i=0$ for every $1\leq i\leq d$, it follows that $m_i=\alpha(\vr{c}_i)$ for every $1\leq i\leq d$. We deduce that $\rho(\vr{b}_i)=\alpha(\vr{b}_{i+1})$ for every $0\leq i<d$. Now, from $s_i=0$ for every $1\leq i\leq d$, we deduce that $n_i=\alpha(\vr{b}_{i+1})$ for every $1\leq i\leq d$. In particular $\beta(\vr{b}_i)=\alpha(\vr{b}_i)$ for every $0\leq i\leq d$. Finally, just observe that from $\sum_{i=1}^d m_i=K$ and $\sum_{i=1}^d n_i=K$, we deduce that $\beta(\vr{y})=\alpha(\vr{y})-2K$. We have proved the lemma.
\end{proof}

\section{Bounded Semantics and Preamplifiers}\label{sec:bounded}
We introduce in this section the $K$-bounded semantics of general programs, where $K$ is a natural number. Intuitively this semantics is obtained by bounding the sum of the counters by $K$. We also introduce the notion of $K$-preamplifiers that provides a way to simulate the $K$-bounded semantics of general programs thanks to checking programs equipped with the classical (unbounded) semantics.

\medskip

More formally, we say that a configuration $\rho$ is $K$-bounded for some $K\in\setN$ if $\sum_{\vr{c}}\rho(\vr{c})\leq K$. Denoting by $\textrm{Conf}_{\leq K}$ the set of $K$-bounded configurations, we define the binary relation $\xrightarrow{M}_{\leq K}$ over $\textrm{Conf}_{\leq K}$ where $M$ is a general program inductively as follows:
$$
\xrightarrow{M}_{\leq K} ~~=~~
\begin{cases}
  (\xrightarrow{\textbf{cmd}})\cap (\textrm{Conf}_{\leq K}\times \textrm{Conf}_{\leq K}) & \text{ if $M=\textbf{cmd}$ is a command}\\
  (\xrightarrow{M_0}_{\leq K})^* \cap (\textrm{Conf}_{\leq K}\times \textrm{Conf}_{\leq K}) & \text{ if $M=\textbf{loop}~M_0$}\\
  \xrightarrow{M_1}_{\leq K};\xrightarrow{M_2}_{\leq K} & \text{ if $M=M_1;M_2$}\\
  \xrightarrow{M_1}_{\leq K}\cup\xrightarrow{M_2}_{\leq K} & \text{ if $M=M_1\textbf{ or }M_2$}
\end{cases}$$

Intuitively $\alpha\xrightarrow{M}_{\leq K}\beta$ for two configurations $\alpha,\beta$ if, and only if, there exists an execution of $M$ such that every visited configuration including $\alpha$ and $\beta$, during the computation is $K$-bounded. The relation $\xrightarrow{M}_{\leq K}$ is called the \emph{$K$-bounded semantics} of $M$.

\medskip

The $K$-bounded semantics can be simulated by checking programs thanks to the so-called $K$-preamplifiers. A $K$-preamplifier~\cite{DBLP:conf/stoc/CzerwinskiLLLM19} for a triple of counters $(\vr{x},\vr{y},\vr{b})$ is a checking program $A$ such that:
\begin{itemize}
\item For every configuration $\beta$ such that $0\xrightarrow{A}\beta$ we have $\vr{y}\geq \vr{b}\vr{x}$ and $\beta(\vr{c})=0$ for every counter $\vr{c}\not\in\{\vr{x},\vr{y},\vr{b}\}$. Moreover, if $\beta$ satisfies $\vr{y}=\vr{b}\vr{x}$ then $\beta(\vr{b})=K$.
\item For every $\ell\geq 1$ there exists a configuration $\beta$ satisfying $0\xrightarrow{A}\beta$,  $\vr{y}=\vr{b}\vr{x}$, and $\vr{x}=\ell$.  
\end{itemize}

\begin{remark}
  A $K$-amplifier~\cite{DBLP:conf/stoc/CzerwinskiLLLM19} for a triple of counters $(\vr{x},\vr{y},\vr{b})$ is a checking program $A$ such that for any configuration $\beta$, we have $0\xrightarrow{A}\beta$ if, and only if, there exists $\ell>0$ such that $\beta(\vr{x},\vr{y},\vr{b})=(\ell,K\ell,K)$ and such that $\beta(\vr{c})=0$ for any counter $\vr{c}\not\in\{\vr{x},\vr{y},\vr{b}\}$. From a $K$-preamplifier, one can compute a $K$-amplifier by introducing some additional counters. We do not introduce $K$-amplifiers in this paper in order to reduce the number of counters used by the simulation. 
\end{remark}

Given a $K$-preamplifier $A$ and a general program $M$ of dimension $d$, we can compute in time $\size{A}+O(d\size{M})$ a checking program $A\rhd M$ such that for any configuration $\beta$, we have $0\xrightarrow{M}_{\leq K}\beta$, if, and only if, $0\xrightarrow{A\rhd M}\beta$. It follows that $K$-preamplifiers provide a way to postpone at the end of an execution test commands of general programs equipped with the $K$-bounded semantics. The size of $A\rhd M$ is $\size{A}+O(d\size{M})$. Concerning the dimension of $A\rhd M$, let us first classify the counters used by $A$. We say that a counter $\vr{c}$ used by the preamplifier $A$ is \emph{unsafe} (for the simulation), if it belongs to $\{\vr{x},\vr{y},\vr{b}\}$ or if it occurs in a test command at the end of $A$, and let us say it is \emph{safe} if it is not unsafe. Then denoting by $u$ and $s$ respectively the number of unsafe and safe counters of $A$, the dimension of $A\rhd M$ is equal to $u+\max(s,d)$.

\medskip

The checking program $A\rhd M$ is obtained as follows. By renaming the counters of $A$, we can assume without loss of generality that the unsafe counters of $A$ are disjoint from the counters used by $M$. Moreover, with such a renaming we can additionally assume that the cardinal of the safe counters of $A$ union the counters used by $M$ is $\max(s,d)$. We assume that $A$ is a checking program of the form $A';\test{\vr{c}_1,\ldots,\vr{c}_n}$ where $A'$ is a test-free program. We can assume that $\vr{b}$ is not in $\{\vr{c}_1,\ldots,\vr{c}_n\}$ since otherwise $K=0$ and in this case $A$ can by replaced by the $0$-preamplifier $\textbf{loop}~\inc{\vr{x}}$. We denote by $\vr{B}$ the counters used by $M$ union $\{\vr{b}\}$. By adding a test command $\test{\vr{c}}$ for some counter $\vr{c}$ used by $M$, we can assume that $M$ starts with a test command (notice that if $M$ is not using any counter, then $M$ is the empty program and the construction $A\rhd M$ can be defined as $M$).

\medskip

We introduce the test-free program $M'$ obtained from $M$ by replacing each increment command $\inc{\vr{c}}$ by $\inc{\vr{c}};\dec{\vr{b}}$, each decrement command $\dec{\vr{c}}$ by $\dec{\vr{c}};\inc{\vr{b}}$, and each test command $\test{\vr{c}}$ by $\textbf{simtest}_{\vr{x},\vr{B},\vr{y}}(\vr{c})$. The checking program $A\rhd M$ is then defined as follows:
\prog{A\rhd M}{
  \State A'
  \State M'
  \Loop
  \State $\dec{\vr{b}}$
  \EndLoop
  \State \test{\vr{x},\vr{y},\vr{b},\vr{c}_1,\ldots,\vr{c}_n}
}

For every configuration $\beta$, we have $0\xrightarrow{M}_{\leq K}\beta$ if, and only if, $0\xrightarrow{A\rhd M}\beta$. A formal proof for a variant construction of $A\rhd M$ (that cannot reuse safe counters of $A$, that introduces several additional test commands, and that uses additional counters) is given in~\cite{DBLP:conf/stoc/CzerwinskiLLLM19}. In the next two paragraphs, we just recall briefly the key ingredients used for the formal proof.

For one direction, assume that $0\xrightarrow{M}_{\leq K}\beta$ for some configuration $\beta$ and let us consider an execution of $M$ from $0$ to $\beta$ such that every visited configuration including $\beta$ is $K$-bounded. Denoting by $m$ the number of times this execution is using a test command, we introduce $\ell=2m$. Since $M$ starts with a test command, it follows that $\ell>0$. We consider an execution of $A'$ that leads to a configuration $\rho$ such that $\rho(\vr{x},\vr{y},\vr{b})=(\ell,K\ell,K)$ and $\rho(\vr{c})=0$ for every counter $\vr{c}\not\in\{\vr{x},\vr{y},\vr{b}\}$. Observe that $\rho$ satisfies $K=\rho(\sum_{\vr{c}\in\vr{B}}\vr{c})$. From the execution of $M$, we derive an execution of $M'$ from $\rho$ to a configuration $\delta$ satisfying $\delta(\vr{x})=0$, $\delta(\vr{y})=0$, and $\delta(\vr{c})=\beta(\vr{c})$ for every counter $\vr{c}$ used by $M$. In fact, every time a test command $\test{\vr{c}}$ is executed, it is simulated by $\textbf{simtest}_{\vr{x},\vr{B},\vr{y}}(\vr{c})$ using Lemma~\ref{lem:exist} that decrements $\vr{x}$ by $2$, and $\vr{y}$ by $2K$. Finally, from $\delta$ we iterate the last loop exactly $\delta(\vr{b})$ times in such a way we get the configuration $\beta$. This configuration can then execute the last test commands of $A\rhd M$.

\medskip

For the other direction, assume that $0\xrightarrow{A\rhd M}\beta$ for some configuration $\beta$ and let us consider an execution of $A\rhd M$ witnessing that property. In $A\rhd M$ the test-free program $A'$ is executed first and the test commands $\test{\vr{c}_1,\ldots,\vr{c}_n}$ of $A$ are postponed at the end of $A\rhd M$. Since those tested counters are no longer used in between, the execution of $A'$ can only produce from the zero configuration a configuration $\rho$ satisfying $\vr{y}\geq \vr{b}\vr{x}$ and $\rho(\vr{c})=0$ for every counter $\vr{c}\not\in\{\vr{x},\vr{y},\vr{b}\}$. We denote by $\delta$ the configuration obtained from $\rho$ after executing $M'$. Let $K'=\rho(\vr{b})$. Assume by contradiction that $\rho$ satisfies the strict constraint $\vr{y}>K'\vr{x}$. Lemma~\ref{lem:forall} shows that $\delta$ satisfies the same strict inequality. In particular $\beta(\vr{y})=\delta(\vr{y})>0$ and the last test command $\test{\vr{y}}$ fails on $\beta$. We get a contradiction. It follows that $\delta$ satisfies the equality $\vr{y}=\vr{b}\vr{x}$ and since $A$ is a $K$-preamplifier, we get $K'=K$ in that case. Lemma~\ref{lem:forall} shows that $\delta$ satisfies $\vr{y}\geq K\vr{x}$. Moreover, since $\beta(\vr{y})=\beta(\vr{x})=0$ we deduce that $\delta$ satisfies the equality $\vr{y}=K\vr{x}$. From Lemma~\ref{lem:forall} we deduce that every execution of $\textbf{simtest}_{\vr{x},\vr{B},\vr{y}}(\vr{c})$ in $M'$ has the same effect as the execution of $\test{\vr{c}};\dec{\vr{y}}^{(2K)};\dec{\vr{x}}^{(2)}$. From the execution of $M'$ we deduce that $\alpha\xrightarrow{M}_{\leq K}\beta$.

\medskip

It follows that the $K$-bounded semantics of general programs can be simulated by checking programs of small size as soon as there exists small size $K$-preamplifiers.

\section{Loop at Most}\label{sec:loop}
In this section we present a way to iterate test-free programs a numbers of times that depends on the valuation of some counters. Given two distinct counters $\vr{c}$ and $\vr{c}'$, and a test-free program $M$, we introduce the following test-free program:
\prog{\textbf{Loop at most $\vr{c}+\vr{c}'$ times $M$}}{
  \Loop
  \State \dec{\vr{c}};\inc{\vr{c}'}
  \EndLoop
  \Loop
  \State \dec{\vr{c}'};\inc{\vr{c}};M
  \EndLoop
}

Let us denote by $\vr{C}$ the counter expression $\vr{c}+\vr{c}'$. In the following lemmas, we assume that $M$ is any test-free program that does not use the counters $\vr{c}$ and $\vr{c}'$.

\begin{lemma}\label{lem:loopatmost}
  For every configuration $\alpha,\beta$, we have:
  $$\alpha\xrightarrow{\textbf{Loop at most $\vr{c}+\vr{c}'$ times $M$}}\beta$$
  if, and only if, there exists $n,\ell\in\setN$ such that $n\leq \alpha(\vr{c})$, $\ell\leq n+\alpha(\vr{c}')$, and such that:
  $$\alpha\xrightarrow{(\dec{\vr{c}};\inc{\vr{c}'})^{(n)};(\dec{\vr{c}'};\inc{\vr{c}};M)^{(\ell)}}\beta$$
  In particular $\ell\leq \alpha(\vr{C})$. If this inequality is an equality then $\beta(\vr{c})=\alpha(\vr{C})$ and $\beta(\vr{c}')=0$.
\end{lemma}
\begin{proof}
  We denote by $N$ and $N_{n,\ell}$ the following test-free programs where $n,\ell\in\setN$:
  \progline{N}{
    \Loop
    \State \dec{\vr{c}};\inc{\vr{c}'}\label{loop:1}
    \EndLoop
    \Loop
    \State \dec{\vr{c}'};\inc{\vr{c}};M\label{loop:2}
    \EndLoop
  }
  
  \prog{N_{n,\ell}}{
    \State $(\dec{\vr{c}};\inc{\vr{c}'})^{(n)}$
    \State $(\dec{\vr{c}'};\inc{\vr{c}};M)^{(\ell)}$
  }
  
  Let $\alpha,\beta$ be two configurations such that $\alpha\xrightarrow{N}\beta$. We fix some execution witnessing that property. Let $n$ be the number of times line~\ref{loop:1} is executed and $\ell$ be the number of times line~\ref{loop:2} is executed. Since each execution of line~\ref{loop:1} decrements $\vr{c}$, we deduce that after executing the first loop we get a configuration $\rho$ satisfying $\rho(\vr{c})=\alpha(\vr{c})-n$ and $\rho(\vr{c}')=\alpha(\vr{c}')+n$. In particular $n\leq \alpha(\vr{c})$. Symmetrically, since each execution of line~\ref{loop:2} decrements $\vr{c}'$, we deduce that $\beta(\vr{c}')=\rho(\vr{c}')-\ell=\alpha(\vr{c}')+n-\ell$ and $\beta(\vr{c})=\rho(\vr{c})+\ell=\alpha(\vr{c})-n+\ell$. In particular $\ell\leq n+\alpha(\vr{c}')$. Observe that $\alpha\xrightarrow{N_{n,\ell}}\beta$. Moreover, we have $\ell\leq \alpha(\vr{c})$. Observe that if $\ell=\alpha(\vr{C})$, then from $\beta(\vr{c}')=\alpha(\vr{c}')+n-\ell$ we get $\beta(\vr{c}')=n-\alpha(\vr{c})$. In particular $n\geq \alpha(\vr{c})$ and with $n\leq \alpha(\vr{c})$ we get $n=\alpha(\vr{c})$. We deduce that $\beta(\vr{c})=\alpha(\vr{C})$ and $\beta(\vr{c}')=0$.

  Conversely, let $\alpha,\beta$ be two configurations such that $\alpha\xrightarrow{N_{n,\ell}}\beta$ for two natural numbers $n,\ell\in\setN$ such that $n\leq\alpha(\vr{c})$ and $\ell\leq n+\alpha(\vr{c}')$. Just observe that from the execution witnessing $\alpha\xrightarrow{N_{n,\ell}}\beta$, we deduce an execution witnessing $\alpha\xrightarrow{N}\beta$ by executing the first loop $n$ times and the second loop $\ell$ times.
\end{proof}

\begin{lemma}\label{lem:loopatmostinc}
  For every configuration $\alpha,\beta$, we have:
  $$\alpha\xrightarrow{\textbf{Loop at most $\vr{c}+\vr{c}'$ times $(\inc{\vr{c}};M)$}}\beta$$
  if, and only if, there exists $n,\ell\in\setN$ such that $n\leq \alpha(\vr{c})$, $\ell\leq n+\alpha(\vr{c}')$, and such that:
  $$\alpha\xrightarrow{(\dec{\vr{c}};\inc{\vr{c}'})^{(n)};(\dec{\vr{c}'};\inc{\vr{c}}^{(2)};M)^{(\ell)}}\beta$$
  In particular $\ell\leq \alpha(\vr{C})$ and $\beta(\vr{C})\leq 2\alpha(\vr{C})$. If one of those two inequalities is an equality then $\ell=\alpha(\vr{C})$, $\beta(\vr{c})=2\alpha(\vr{C})$, and $\beta(\vr{c}')=0$.
\end{lemma}
\begin{proof}
  The proof is very similar to the proof of Lemma~\ref{lem:loopatmost}.
\end{proof}

\begin{lemma}\label{lem:loopatmostdec}
  For every configuration $\alpha,\beta$, we have:
  $$\alpha\xrightarrow{\textbf{Loop at most $\vr{c}+\vr{c}'$ times $(\dec{\vr{c}'};M)$}}\beta$$
  if, and only if, there exists $n,\ell\in\setN$ such that $n\leq \alpha(\vr{c})$, $\ell\leq \frac{n+\alpha(\vr{c}')}{2}$, and such that:
  $$\alpha\xrightarrow{(\dec{\vr{c}};\inc{\vr{c}'})^{(n)};(\dec{\vr{c}'};\inc{\vr{c}};\dec{\vr{c}'};M)^{(\ell)}}\beta$$
  In particular $\ell\leq \frac{\alpha(\vr{C})}{2}$. If this inequality is an equality then $\beta(\vr{c})=\frac{\alpha(\vr{C})}{2}$, and $\beta(\vr{c}')=0$.
\end{lemma}
\begin{proof}
  We denote by $N$ and $N_{n,\ell}$ the following test-free programs where $n,\ell\in\setN$:
  \progline{N}{
    \Loop
    \State \dec{\vr{c}};\inc{\vr{c}'}\label{loop:div1}
    \EndLoop
    \Loop
    \State \dec{\vr{c}'};\inc{\vr{c}};\dec{\vr{c}'};M\label{loop:div2}
    \EndLoop
  }
  
  \prog{N_{n,\ell}}{
    \State $(\dec{\vr{c}};\inc{\vr{c}'})^{(n)}$
    \State $(\dec{\vr{c}'};\inc{\vr{c}};\dec{\vr{c}'};M)^{(\ell)}$
  }
  
  Let $\alpha,\beta$ be two configurations such that $\alpha\xrightarrow{N}\beta$. We fix some execution witnessing that property. Let $n$ be the number of times line~\ref{loop:1} is executed and $\ell$ be the number of times line~\ref{loop:2} is executed. Since each execution of line~\ref{loop:div1} decrements $\vr{c}$, we deduce that after executing the first loop we get a configuration $\rho$ satisfying $\rho(\vr{c})=\alpha(\vr{c})-n$ and $\rho(\vr{c}')=\alpha(\vr{c}')+n$. In particular $n\leq \alpha(\vr{c})$. Symmetrically, since each execution of line~\ref{loop:div2} decrements $\vr{c}'$ two times, we deduce that $\beta(\vr{c}')=\rho(\vr{c}')-2\ell=\alpha(\vr{c}')+n-2\ell$ and $\beta(\vr{c})=\rho(\vr{c})+\ell=\alpha(\vr{c})-n+\ell$. In particular $\ell\leq \frac{n+\alpha(\vr{c}')}{2}$. Observe that $\alpha\xrightarrow{N_{n,\ell}}\beta$. Moreover, we have $\ell\leq \frac{\alpha(\vr{C})}{2}$. Observe that if $\ell=\frac{\alpha(\vr{C})}{2}$, then from $\beta(\vr{c}')=\alpha(\vr{c}')+n-2\ell$ we get $\beta(\vr{c}')=n-\alpha(\vr{c})$. In particular $n\geq \alpha(\vr{c})$ and with $n\leq \alpha(\vr{c})$ we get $n=\alpha(\vr{c})$. We deduce that $\beta(\vr{c})=\frac{\alpha(\vr{C})}{2}$ and $\beta(\vr{c}')=0$.

  Conversely, let $\alpha,\beta$ be two configurations such that $\alpha\xrightarrow{N_{n,\ell}}\beta$ for two natural numbers $n,\ell\in\setN$ such that $n\leq\alpha(\vr{c})$ and $\ell\leq \frac{n+\alpha(\vr{c}')}{2}$. Just observe that from the execution witnessing $\alpha\xrightarrow{N_{n,\ell}}\beta$, we deduce an execution witnessing $\alpha\xrightarrow{N}\beta$ by executing the first loop $n$ times and the second loop $\ell$ times.
\end{proof}

\begin{lemma}\label{lem:loopatmostdecinc}
  For every configuration $\alpha,\beta$, we have:
  $$\alpha\xrightarrow{\textbf{Loop at most $\vr{c}+\vr{c}'$ times $(\dec{\vr{c}'};\inc{\vr{c}};M)$}}\beta$$
  if, and only if, there exists $n,\ell\in\setN$ such that $n\leq \alpha(\vr{c})$, $\ell\leq \frac{n+\alpha(\vr{c}')}{2}$, and such that:
  $$\alpha\xrightarrow{(\dec{\vr{c}};\inc{\vr{c}'})^{(n)};((\dec{\vr{c}'};\inc{\vr{c}})^{(2)};M)^{(\ell)}}\beta$$
  In particular $\ell\leq \frac{\alpha(\vr{C})}{2}$. If this inequality is an equality then $\beta(\vr{c})=\alpha(\vr{C})$, and $\beta(\vr{c}')=0$.
\end{lemma}
\begin{proof}
  The proof is very similar to the proof of Lemma~\ref{lem:loopatmostdec}.
\end{proof}

\section{Implementing $\reduce_d$}\label{sec:reduce}
In this section, we introduce a test-free program $\textbf{evalF}_d$ that implements the function $\reduce_d$. This program is using counters $\vr{x},\vr{x}',\vr{x}_{1},\ldots,\vr{x}_{d+1},\vr{y},\vr{b},\vr{b}',\vr{c}_0,\vr{c}_0',\vr{c}_1,\ldots,\vr{c}_d$. We say that a configuration $\rho$ is \emph{good} if it satisfies $\vr{x}'=\vr{y}=\vr{b}'=\vr{c}_0'=0$, $\vr{x}>0$, $\vr{b}=2^{\vr{c}_0}$, $\vr{x}_1=2^{\vr{c}_0}\vr{x}$ and $\vr{x}_{i}=2^{\vr{c}_{i-1}}\vr{x}_{i-1}$ for every $i\in \{2,\ldots,d+1\}$. We say that a good configuration $\rho$ encodes a pair $(v,n)\in\setN^d\times\setN$ if $\rho(\vr{c}_1,\ldots,\vr{c}_d)=v$ and $\rho(\vr{c}_0)=n$.

\begin{remark}
  The counters $\vr{x}$ and $\vr{x}'$ should have been denoted as $\vr{x}_0$ and $\vr{x}_0'$ to simplify a little bit the definition of good configurations. However, since those two counters appear many times in the sequel, we prefer the notation without indexes. Moreover, this notation match the one used for preamplifiers. In fact, in the next section we introduce an Ackermannian preamplifier for the triple of counters $(\vr{x},\vr{y},\vr{b})$ without any variable renaming.
\end{remark}

\medskip

We introduce the counter expressions $\vr{C}_0$, $\vr{X}$ and $\vr{B}$ defined as $\vr{c}_0+\vr{c}_0'$, $\vr{x}+\vr{x}'$, and $\vr{b}+\vr{b}'$ respectively.

\medskip

\newcommand{\ibad}[1]{\operatorname{index}(#1)}

Intuitively, when the computation of $\textbf{evalF}_d$ is correct from a good configuration that encodes a pair $(v,n)$ with $v\not=\zero_d$ then the computation terminates on a good configuration that encodes $\reduce_d(v,n)$. When the computation is incorrect, we obtain a so-called bad configuration. Moreover, from a bad configuration any computation of $\textbf{evalF}_d$ leads to a bad configuration. Formally, we say that a configuration is \emph{$i$-bad} for some $i\in\{2,\ldots,d+1\}$ if it satisfies $\vr{x}+\vr{x}'>0$, $\vr{x}_{i}+\vr{y}>2^{\vr{c}_{i-1}}(\vr{x}_{i-1}+\vr{y})$, and $\vr{x}_{j}+\vr{y}=2^{\vr{c}_{j-1}}(\vr{x}_{j-1}+\vr{y})$ for every $j\in \{i+1,\ldots,d+1\}$. A configuration is \emph{$1$-bad} if is satisfies $\vr{x}'=\vr{y}=0$, $\vr{x}>0$, $\vr{b}+\vr{b}'<2^{\vr{c}_0+\vr{c}_0'}$, $\vr{x}_1=2^{\vr{c}_0+\vr{c}_0'}\vr{x}$, and $\vr{x}_{j}=2^{\vr{c}_{j-1}}\vr{x}_{j-1}$ for every $j\in \{2,\ldots,d+1\}$. We say that a configuration $\alpha$ is \emph{bad} if it is $i$-bad for some $i$. A configuration that is \emph{good} or \emph{bad} is said to be \emph{conform}. We associate to any conform configuration $\alpha$ its index of badness $\ibad{\alpha}$ defined by $\ibad{\alpha}=i$ if $\alpha$ is $i$-bad, and by $\ibad{\alpha}=0$ if $\alpha$ is good. Intuitively, from any conform configuration $\alpha$, the computation of $\textbf{evalF}_d$ can only produce conform configurations $\beta$ such that $\ibad{\beta}\geq \ibad{\alpha}$. In particular, if $\beta$ is good, then $\alpha$ is good as well.

\medskip

The program $\textbf{evalF}_{d}$ is defined as a non-deterministic choice of programs $\textbf{evalF}_{d,p}$ with $p\in\{1,\ldots,d\}$. Intuitively from a good configuration $\alpha$ that encodes a pair $(v,n)$ with $v\not=\zero_d$, the computation of $\textbf{evalF}_{d,p}$ can lead to a good configuration only if $p$ is the minimal index such that $v[p]>0$. In that case the good configuration produced at the end of the computation encodes $\reduce_d(v,n)$. 

\medskip

We introduce the following test-free program for $p\in\{1,\ldots,d\}$. The first loop at line~\ref{loop:first} intuitively transfer $\min{\vr{x}_1,\ldots,\vr{x}_d}$ into $\vr{y}$. Thanks to Lemma~\ref{lem:loopatmost}, the second loop at line~\ref{loop:k} transfer back from $\vr{y}$ the value of $\vr{X}$. Notice that $\vr{x}_p$ is incremented twice during each iteration of the loop. The loop at line~\ref{loop:k} is interpreted thanks to Lemma~\ref{lem:loopatmost} if $p>1$ and Lemma~\ref{lem:loopatmostinc} if $p=1$ as a loop that is iterated $\vr{C}_0$ times, the loop at line~\ref{loop:x} is interpreted thanks to Lemma~\ref{lem:loopatmostdec} as a loop that is iterated $\vr{X}$ times and that divides $\vr{X}$ by $2$, and finally the loop at line~\ref{loop:v} is interpreted thanks to Lemma~\ref{lem:loopatmost} as a loop that is iterated $\vr{B}$ times.
\progline{\textbf{evalF}_{d,p}}{
  \State \dec{\vr{c}_p};\inc{\vr{c}_{p-1}}
  \Loop\label{loop:first}
  \State \dec{\vr{x}_1,\ldots,\vr{x}_{d+1}};\inc{\vr{y}}\label{l:Er}
  \EndLoop
  \Loopatmost{\vr{x}}{\vr{x}'}\label{loop:s}
    \State \dec{\vr{y}};\inc{\vr{x}_p};\inc{\vr{x}_1,\ldots,\vr{x}_{d+1}}\label{l:Es}
  \EndLoop
  \Loopatmost{\vr{c}_0}{\vr{c}_0'}\label{loop:k}
    \State \inc{\vr{c}_{p-1}}\label{l:Ek}
    \Loopatmost{\vr{x}}{\vr{x}'}\label{loop:x}
      \State $\dec{\vr{x}'}$\label{l:Em}
      \Loopatmost{\vr{b}}{\vr{b}'}\label{loop:v}
        \State \dec{\vr{y}};$\inc{\vr{x}_p}$;\inc{\vr{x}_p,\ldots,\vr{x}_{d+1}}\label{l:En}
      \EndLoop
    \EndLoop
  \EndLoop
  \State $\textbf{updateB}_{d,p}$  
}

Where $\textbf{updateB}_{d,p}$ is the empty program if $p>1$ and the following one if $p=1$. Notice that the loop at line~\ref{loop:upc} can be interpreted thanks to Lemma~\ref{lem:loopatmostdecinc} as a loop that is iterated $\frac{\vr{C}_0}{2}$ times, and the loop at line~\ref{loop:upb} is interpreted thanks to Lemma~\ref{lem:loopatmostinc} as a loop that multiply by $2$ the value of $\vr{B}$.
\progline{\textbf{updateB}_{d,1}}{
  \State \inc{\vr{c}_0}\label{l:upc}
  \Loopatmost{\vr{c}_0}{\vr{c}_0'}\label{loop:upc}
  \State \dec{\vr{c}_0'};\inc{\vr{c}_0}\label{l:ups}
  \Loopatmost{\vr{b}}{\vr{b'}}\label{loop:upb}
  \State \inc{\vr{b}}\label{l:upk}
  \EndLoop
  \EndLoop
  \State \dec{\vr{c}_0}\label{l:downc}
}

\medskip

 We first provide two lemmas describing the behaviour of $\textbf{updateB}_{d,1}$.
\begin{lemma}\label{lem:updateF}
  Assume that $\alpha\xrightarrow{\textbf{updateB}_{d,1}}\beta$ for any two configurations $\alpha,\beta$. Then $\beta(\vr{C}_0)=\alpha(\vr{C}_0)$ and we have:
  \begin{equation}\label{equ:update}
    \beta(\vr{B})\leq \alpha(2^{\frac{\vr{C}_0+1}{2}}\vr{B})
  \end{equation}
  If the previous inequality is an equality and $\alpha(\vr{B})>0$ then $\beta$ satisfies $\vr{b}'=0$ and $\vr{c}_0'=0$.
\end{lemma}
\begin{proof}
  Let $s$ be the number of times line~\ref{l:ups} is executed. Lemma~\ref{lem:loopatmostdecinc} shows that $s\leq \frac{1+\alpha(\vr{C}_0)}{2}$ (the $1$ in the expression comes form the increment command at line~\ref{l:downc}). Let $\rho_{j-1}$ for $j\in\{1,\ldots,s\}$ be the configuration just before executing that line the $j$th time. We also denote by $\rho_s$ the configuration just before the execution of line~\ref{l:downc}. Lemma~\ref{lem:loopatmostinc} shows that $\rho_j(\vr{B})\leq \rho_{j-1}(\vr{B})$ for every $j\in\{1,\ldots,s\}$. We deduce that the inequality~(\ref{equ:update}) holds. If this inequality is an equality and $\alpha(\vr{B})>0$, then $s=\frac{1+\alpha(\vr{C}_0)}{2}$. Lemma~\ref{lem:loopatmostdecinc} shows that $\beta(\vr{c}_0')=0$. Notice that since the inequality~(\ref{equ:update}) is an equality, then the inequality $\rho_j(\vr{B})\leq 2\rho_{j-1}(\vr{B})$ is also an equality for every $j\in\{1,\ldots,k\}$. In particular $\rho_k(\vr{b}')=0$ from Lemma~\ref{lem:loopatmostinc}. It follows that $\beta(\vr{b}')=0$.
\end{proof}

\begin{lemma}\label{lem:updateE}
  Let $\alpha$ be any configuration satisfying $\alpha(\vr{C}_0)$ is odd. Then $\alpha\xrightarrow{\textbf{updateB}_{d,1}}\beta$ where $\beta$ is the configuration defined for every counter $\vr{c}$ as follows:
  $$\beta(\vr{c})=
  \begin{cases}
    \alpha(\vr{C}_0) & \text{ if }\vr{c}=\vr{c}_0\\
    0 & \text{ if }\vr{c}=\vr{c}_0'\\
    \alpha(2^{\frac{\vr{C}_0+1}{2}}\vr{B}) & \text{ if }\vr{c}=\vr{b}\\
    0 & \text{ if }\vr{c}=\vr{b}'\\
    \alpha(\vr{c}) & \text{ otherwise}
  \end{cases}$$
\end{lemma}
\begin{proof}
  Just observe that we can apply Lemma~\ref{lem:loopatmostdecinc} and Lemma~\ref{lem:loopatmostinc}.
\end{proof}

\begin{lemma}\label{lem:evalF}
  Let $\alpha$ be a conform configuration and $\beta$ be any configuration such that $\alpha\xrightarrow{\textbf{evalF}_{d,p}}\beta$ with $p\in\{1,\ldots,d\}$. Then $\beta$ is conform and $\ibad{\beta}\geq \ibad{\alpha}$. In particular, if $\beta$ is good then $\alpha$ is good. In that case, $\alpha$ satisfies $\vr{b}$ divides $\vr{x}$, denoting by $(v,n)$ the pair encoded by $\alpha$ then $v\not=\zero_d$, $p$ is the minimal index such that $v[p]>0$, $\reduce_d(v,n)$ is encoded by $\beta$, and $\beta(\vr{x}_{d+1})=\alpha(\vr{x}_{d+1})$.
\end{lemma}
\begin{proof}  
  Let $r,s,k,m$ and $n$ be the number of times lines~\ref{l:Er}, \ref{l:Es}, \ref{l:Ek}, \ref{l:Em}, and \ref{l:En} are executed respectively. Observe that $r\leq \min\{\alpha(\vr{x}_1),\ldots,\alpha(\vr{x}_{d+1})\}$ since every time line~\ref{l:Er} is executed, the counters $\vr{x}_1,\ldots,\vr{x}_{d+1}$ are decremented, $s\leq \alpha(\vr{X})$ thanks to Lemma~\ref{lem:loopatmost}, $k\leq \alpha(\vr{C}_0)$, thanks to Lemma~\ref{lem:loopatmost} if $p>1$ and Lemma~\ref{lem:loopatmostinc} if $p=1$ (in fact in that case $\vr{c}_{p-1}=\vr{c}_0$),  and $n\leq \alpha(\vr{B})m$ thanks to Lemma~\ref{lem:loopatmost}.

  Let us prove that $\beta(\vr{X})>0$. Observe that $\alpha(\vr{X})>0$ since $\alpha$ is conform. Moreover, line~\ref{l:Em} is the unique line that may prevent $\beta(\vr{X})>0$ holding. However, each time this line is executed, by definition of \textbf{loop at most}, notice that $\vr{x}$ is incremented just before. It follows that $\beta(\vr{X})>0$.
  
  For each $j\in\{1,\ldots,k\}$, let $m_j$ be the number of times line~\ref{l:Em} is executed during the $j$th execution of the loop at line~\ref{loop:k}. We denote by $\alpha_{j-1}$ the configuration just before the $j$th execution of line~\ref{l:Em}, and by $\alpha_k$ the configuration just before executing $\textbf{updateB}_{d,p}$. Notice that $\alpha_0(\vr{X})=\alpha(\vr{X})$, $2m_j\leq \alpha_{j-1}(\vr{X})$, and $\alpha_j(\vr{X})=\alpha_{j-1}(\vr{X})-m_j$ for every $j\in\{1,\ldots,k\}$. Let us denote by $h_j$ the natural number such that $\alpha_{j-1}(\vr{X})=2m_j+h_j$. By induction on $i\in\{0,\ldots,k\}$, we get the following equality:
  $$\sum_{j=1}^i m_j = \alpha(\vr{X})[1-2^{-i}]-\sum_{j=1}^i 2^{j-i-1}h_j$$
  In particular with $i=k$, since $m=\sum_{j=1}^k m_j$, we get the following equality by introducing the non negative rational number $h=\sum_{j=1}^k 2^{j-k-1}h_j$:
  \begin{equation}\label{eq:m}
    2^{\alpha(\vr{C}_0)}m = -2^{\alpha(\vr{C}_0)}h-\frac{1}{2^k}\alpha(\vr{X})(2^{\alpha(\vr{C}_0)}-2^k)+\alpha((2^{\vr{C}_0}-1)\vr{X})
  \end{equation}
  
  \medskip
  
  Assume first that $\alpha$ is $i$-bad for some $i\in\{p+2,\ldots,d+1\}$. In that case $\alpha$ satisfies $\vr{x}_{i}+\vr{y}>2^{\vr{c}_{i-1}}(\vr{x}_{i-1}+\vr{y})$, and $\vr{x}_{j}+\vr{y}=2^{\vr{c}_{j-1}}(\vr{x}_{j-1}+\vr{y})$ for every $j\in\{i+1,\ldots, d+1\}$. Notice that $\beta$ also satisfies the same constraints since $\vr{x}_j+\vr{y}$ for $j\in\{i,\ldots,d+1\}$, and $\vr{c}_j$ for $j\in\{i,\ldots,d\}$ are invariant. Hence $\beta$ is $i$-bad in that case.

  So we can assume that $\alpha$ is not $i$-bad for every $i\in\{p+2,\ldots,d+1\}$. Since $\alpha$ is conform, we deduce that $\alpha$ satisfies $\vr{x}_{p+1}+\vr{y}\geq 2^{\vr{c}_p}(\vr{x}_p+\vr{y})$, and $\vr{x}_{j}+\vr{y}=2^{\vr{c}_{j-1}}(\vr{x}_{j-1}+\vr{y})$ for every $j\in\{p+2,\ldots,d+1\}$. Since $\vr{x}_j+\vr{y}$ for $j\in\{p+1,\ldots,d+1\}$, and $\vr{c}_j$ for $j\in\{p+1,\ldots,d\}$ are invariant, we deduce that $\beta$ satisfies $\vr{x}_{j}+\vr{y}=2^{\vr{c}_{j-1}}(\vr{x}_{j-1}+\vr{y})$ for every $j\in\{p+2,\ldots,d+1\}$.

  Observe that we have the following equalities:
  \begin{align*}
    \beta(\vr{x}_{p+1}+\vr{y})&=\alpha(\vr{x}_{p+1}+\vr{y})\\
    \beta(\vr{c}_p)&=\alpha(\vr{c}_p)-1\\
    \beta(\vr{x}_p+\vr{y})&=\alpha(\vr{x}_p+\vr{y})+s+n\\
    \beta(\vr{y})&=\alpha(\vr{y})+r-(s+n)
  \end{align*}
  From those equalities, we derive:
  \begin{align*}
    \beta(\vr{x}_{p+1}+\vr{y}-2^{\vr{c}_p}(\vr{x}_p+\vr{y}))
    &=\alpha(\vr{x}_{p+1}+\vr{y}-2^{\vr{c}_p}(\vr{x}_p+\vr{y}))\\
    &+2^{\alpha(\vr{c}_p)-1}(\alpha(\vr{x}_p)-r)\\
    &+2^{\alpha(\vr{c}_p)-1}\beta(\vr{y})
  \end{align*}
  Notice that the right hand-side of the last equality is a sum of three natural numbers. It means that if one of those numbers is strictly positive, then $\beta$ is $(p+1)$-bad. So, we can assume that those three numbers are zero. It means that $\alpha$ and $\beta$ satisfies $\vr{x}_{p+1}+\vr{y}=2^{\vr{c}_p}(\vr{x}_p+\vr{y})$, $r=\alpha(\vr{x}_p)$, and $\beta(\vr{y})=0$. In particular $\alpha$ is not $(p+1)$-bad.
  
  Assume by contradiction that $\alpha$ is $i$-bad for some $i\in\{2,\ldots,p\}$. It means that $\alpha$ satisfies $\vr{x}_{i}+\vr{y}>2^{\vr{c}_{i-1}}(\vr{x}_{i-1}+\vr{y})$ and $\vr{x}_{j}+\vr{y}=2^{\vr{c}_{j-1}}(\vr{x}_{j-1}+\vr{y})$ for every $j\in\{i+1,\ldots,d+1\}$. In particular $\alpha$ satisfies $\vr{x}_{i}>\vr{x}_{i-1}$ and $\vr{x}_{j}\geq \vr{x}_{j-1}$ for every $j\in\{i+1,\ldots,d\}$. We deduce that $\alpha$ satisfies $\vr{x}_i<\vr{x}_p$. It follows that $r\leq \min\{\alpha(\vr{x}_1),\ldots,\alpha(\vr{x}_{d+1})\}<\alpha(\vr{x}_p)=r$ and we get a contradiction.

  \medskip

  Since $\alpha$ is not $i$-bad for every $i\in\{2,\ldots,d+1\}$, we deduce that $\alpha$ is $1$-bad or good. In both cases, $\alpha$ satisfies $\vr{x}'=\vr{y}=0$, $\vr{B}\leq 2^{\vr{C}_0}$, $\vr{x}_1=2^{\vr{C}_0}\vr{x}$, and $\vr{x}_i=2^{\vr{c}_{i-1}}\vr{x}_{i-1}$ for every $i\in\{2,\ldots,d+1\}$. In particular $\alpha(\vr{x}_{d+1})\geq \cdots\geq \alpha(\vr{x})>0$. Notice that from $\alpha(\vr{y})=0$ and $\beta(\vr{y})=0$, we deduce from $\beta(\vr{y})=\alpha(\vr{y})+r-(s+n)$ that $r=s+n$.

  Assume by contradiction that $\alpha(\vr{c}_{i-1})>0$ for some $i\in\{2,\ldots,p\}$. Since $\alpha$ satisfies $\vr{x}_{i}=2^{\vr{c}_{i-1}}\vr{x}_{i-1}$, $\vr{c}_{i-1}>0$, and $\vr{x}_{i-1}>0$, we deduce that $\alpha(\vr{x}_i)>\alpha(\vr{x}_{i-1})$. In particular, $\alpha(\vr{x}_p)\geq \alpha(\vr{x}_i)>\alpha(\vr{x}_{i-1})$. Like in the previous paragraph we get a contradiction with $r\leq \min\{\alpha(\vr{x}_1),\ldots,\alpha(\vr{x}_{d+1})\}<\alpha(\vr{x}_p)=r$. Therefore $\alpha(\vr{c}_{i-1})=0$ for every $i\in\{2,\ldots,p\}$. We deduce that $\alpha(\vr{x}_1)=\cdots=\alpha(\vr{x}_p)$. It follows from $r=\alpha(\vr{x}_p)$ and $r=s+n$ that $\alpha(\vr{x}_1)=s+n$.

  \medskip

  Observe that we have (by developing the equalities and replacing $2^{\alpha(\vr{C}_0)}m$ by using the equality~\ref{eq:m}):
  \begin{align*}
    n &= -(\alpha(\vr{B})m-n)-\alpha(2^{\vr{B_0}}-\vr{B})m -2^{\alpha(\vr{C}_0)}h-\frac{1}{2^k}\alpha(\vr{X})(2^{\alpha(\vr{C}_0)}-2^k)+\alpha((2^{\vr{C}_0}-1)\vr{X}))\\
    s&= -(\alpha(\vr{X})-s)+\alpha(\vr{X})
  \end{align*}
  It follows that we have (we use the fact that $\alpha(\vr{x}_1-2^{\vr{C}_0}\vr{X})=0$):
  \begin{align*}
    0=\alpha(\vr{x}_1)-(s+n)
    =&(\alpha(\vr{B})m-n)\\
    &+(2^{\alpha(\vr{C}_0)}-\alpha(\vr{B}))m\\
    &+(\alpha(\vr{X})-s)\\
    &+2^{\alpha(\vr{C}_0)}h\\
    &+\frac{1}{2^k}\alpha(\vr{X})(2^{\alpha(\vr{C}_0)}-2^k)
    \end{align*}
  Since each term of the right hand side of the previous equality are non negative, we deduce that all the terms are zero. It means that $n=\alpha(\vr{B})m$, $(2^{\alpha(\vr{C}_0)}-\alpha(\vr{B}))m=0$, $s=\alpha(\vr{X})$, $h=0$, $2^{\alpha(\vr{C}_0)}=2^{k}$ (for the last equality, we use $\alpha(\vr{X})>0$). Since $k\leq \alpha(\vr{C}_0)$ we derive from $2^{\alpha(\vr{C}_0)}=2^{k}$ that $k=\alpha(\vr{C}_0)$. By replacing $h$ by $0$, and $k$ by $\alpha(\vr{C}_0)$ in equation~\ref{eq:m}, we derive $m=\alpha((1-2^{-\vr{C}_0})\vr{X})$.
  
  Assume that $\alpha$ is $1$-bad. In that case $\alpha(\vr{B})<2^{\alpha(\vr{C}_0)}$. It follows from $(2^{\alpha(\vr{C}_0)}-\alpha(\vr{B}))m=0$ that $m=0$. From $m=\alpha((1-2^{-\vr{C}_0})\vr{X})$ and $\alpha(\vr{X})>0$ we get $\alpha(\vr{C}_0)=0$. So, from $\alpha(\vr{B})<2^{\alpha(\vr{C}_0)}$ we get $\alpha(\vr{B})=0$. Observe that in that case $\beta$ is $1$-bad (the case $p=1$ is obtained thanks to Lemma~\ref{lem:updateF}). So, we can now assume that $\alpha$ is not $1$-bad.

  Since $\alpha$ is conform and not bad, we deduce that $\alpha$ is good. Hence $\alpha$ satisfies $\vr{x}'=\vr{c}_0'=\vr{b}'=\vr{y}=0$ and $\vr{b}=2^{\vr{c}_0}$. It follows from $n=\alpha(\vr{B})m$ that $n=\alpha((2^{\vr{C}_0}-1)\vr{X})$. Let $\rho$ be the configuration we obtain just before executing $\textbf{updateB}_{d,p}$. Since $h=0$, then $h_k=0$. It follows that $\rho(\vr{x}')=0$. As $s=\alpha(\vr{X})$, Lemma~\ref{lem:loopatmost} shows that $\rho(\vr{x}')=0$. Moreover, as $n=\alpha(\vr{B})m$, Lemma~\ref{lem:loopatmost} shows that $\rho(\vr{b}')=0$. It follows that $\rho(\vr{x})=\alpha(\vr{x})-m=\alpha(2^{-\vr{c}_0}\vr{x})$. Notice that $\rho(\vr{x}_i)$ with $i\in\{1,\ldots,p-1\}$ is equal to $\alpha(\vr{x}_i)-r+s=s=\alpha(\vr{x})$. We also have $\rho(\vr{x}_p)=\alpha(\vr{x}_p)-r+2s+2n=2\alpha(\vr{x}_p)$. Notice that $\rho(\vr{x}_i)$ with $i\in\{p+1,\ldots,d+1\}$ is equal to $\alpha(\vr{x}_i)$ since $\vr{x}_i+\vr{y}$ is an invariant and $\alpha(\vr{y})=\rho(\vr{y})=0$.

  It follows that for every counter $\vr{c}$, we have:
  $$\rho(\vr{c}) =
  \begin{cases}
    0 & \text{ if }\vr{c}\in\{\vr{x}',\vr{c}_0',\vr{b}',\vr{y}\}\\
    \alpha(2^{-\vr{c}_0}\vr{x}) & \text{ if }\vr{c}=\vr{x}\\
    \alpha(\vr{x}) & \text{ if $\vr{c}=\vr{x}_i$ with $i\in\{1,\ldots,p-1\}$}\\
    2\alpha(\vr{x}_p) & \text{ if $\vr{c}=\vr{x}_p$}\\
    \alpha(\vr{c}_p)-1 & \text{ if $\vr{c}=\vr{c}_p$}\\
    \alpha(1+\vr{c}_{p-1}+\vr{c}_{0}) & \text{ if $\vr{c}=\vr{c}_{p-1}$}\\
    \alpha(\vr{c}) & \text{ otherwise}
  \end{cases}$$
  Notice that if $p>1$ then $\alpha(\vr{c}_{p-1})=0$ and $\beta=\rho$. In that case $\beta$ is good and satisfies the lemma.

  If $p=1$, notice that $\rho(\vr{c}_0)=1+2\alpha(\vr{c}_0)$. From Lemma~\ref{lem:updateF} we deduce that $\beta$ is either $1$-bad or good depending depending if the inequality of that lemma is strict or an equality. Notice that if the inequality is an equality, the configuration $\beta$ is good and it encodes
  $\reduce_d(v,n)$.
\end{proof}

\begin{lemma}\label{lem:evalE}
  Let $\alpha$ be a good configuration satisfying $\vr{b}$ divides $\vr{x}$ and that encodes a pair $(v,n)$ with $v\not=\zero_d$, and let $p$ is the minimal index such that $v[p]>0$. There exists a good configuration $\beta$ that encodes $\reduce_d(v,n)$ satisfying $\beta(\vr{x}_{d+1})=\alpha(\vr{x}_{d+1})$, and such that $\alpha\xrightarrow{\textbf{evalF}_{d,p}}\beta$.
\end{lemma}
\begin{proof}
  We are going to prove that there exists an execution of $\textbf{evalF}_{d,p}$ from $\alpha$ such that lines~\ref{l:Er}, \ref{l:Es}, \ref{l:Ek} are executed $r,s,k$ times with:
  \begin{align*}
    r&=\alpha(\vr{x}_1)\\
    s&=\alpha(\vr{x})\\
    k&=\alpha(\vr{c}_0)
  \end{align*}
  
  We also introduce the configurations that will appear along the execution of loop at line~\ref{loop:k}. To do so, we introduce the configurations $\alpha_0,\ldots,\alpha_k$ defined as follows for every counter $\vr{c}$ and for every $j\in\{0,\ldots,k\}$:
  $$
  \alpha_j(\vr{c})=
  \alpha\left(\begin{cases}
    \frac{\vr{x}}{2^j} & \text{ if $\vr{c}=\vr{x}$}\\
    \vr{x} & \text{ if $\vr{c}=\vr{x}_i$ with $i\in\{1,\ldots,p-1\}$}\\
    2\vr{b}\vr{x}+2\vr{x}(1-2^{\vr{c}_0-j}) & \text{ if $\vr{c}=\vr{x}_p$}\\
    \vr{x}_i+\vr{x}(1-2^{\vr{c}_0-j}) & \text{ if $\vr{c}=\vr{x}_i$ with $i\in\{p+1,\ldots,d+1\}$}\\
    \vr{x}(2^{\vr{c}_0-j}-1) & \text{ if $\vr{c}=\vr{y}$}\\
    \vr{c}_p-1 & \text{ if $\vr{c}=\vr{c}_p$}\\
    \vr{c}_{p-1}+1+j & \text{ if $\vr{c}=\vr{c}_{p-1}$ and $p>1$}\\
    \vr{c}_0-j & \text{ if $\vr{c}=\vr{c}_0'$}\\
    j & \text{ if $\vr{c}=\vr{c}_0$ and $p>1$}\\
    1+2j & \text{ if $\vr{c}=\vr{c}_0$ and $p=1$}\\
    \vr{c} & \text{otherwise}
  \end{cases}
  \right)
  $$

  Since $\alpha$ is good, it follows that $\alpha(\vr{x}_i)\geq  r$ for every $i\in\{1,\ldots,d+1\}$, and $r\geq \alpha(\vr{x})$. In particular line~\ref{l:Er} and~\ref{l:Es} can be effectively executed $r$ and $s$ times respectively. Notice that $\alpha_0$ is the configuration we obtain after executing those two loops and the hidden first loop of the $\textbf{loop at most}$ at line~\ref{loop:k} exactly $\alpha(\vr{x})$ times.

  We introduce the test-free program $N$ corresponding to the subprogram between line~\ref{l:Ek} and line~\ref{l:En} (including those those lines) and prefixed by $\dec{\vr{c}_0'};\inc{\vr{c}_0}$. Let us prove by induction on $j\in\{1,\ldots,k\}$ that $\alpha_{j-1}\xrightarrow{N}\alpha_j$. To do so, we are going to prove that there exists an execution of $N$ such that lines~\ref{l:Em}, \ref{l:En} are executed $m_j,n_j$ times with:
  \begin{align*}
    m_j &=\alpha_{j-1}(\frac{\vr{x}}{2})\\
    n_j&=\alpha_{j-1}(m_j\vr{b})
  \end{align*}
  Notice that those three numbers are natural numbers since $2^k$ divides $\alpha(\vr{x})$. This execution is obtained thanks to Lemma~\ref{lem:loopatmost} applied on line~\ref{l:En} in order to execute maximally the corresponding loop, and by applying Lemma~\ref{lem:loopatmostdec} applied on line~\ref{lem:loopatmostdec}. We have proved that $\alpha_{j-1}\xrightarrow{M}\alpha_j$. In particular, by executing $k$ times loop~\ref{loop:k}, we get the configuration $\alpha_k$. Notice that this configuration satisfies for every counter $\vr{c}$:
  $$
  \alpha_k(\vr{c})=
  \alpha\left(\begin{cases}
    \frac{\vr{x}}{\vr{b}} & \text{ if $\vr{c}=\vr{x}$}\\
    \vr{x} & \text{ if $\vr{c}=\vr{x}_i$ with $i\in\{1,\ldots,p-1\}$}\\
    2\vr{b}\vr{x} & \text{ if $\vr{c}=\vr{x}_p$}\\
    \vr{x}_i & \text{ if $\vr{c}=\vr{x}_i$ with $i\in\{p+1,\ldots,d+1\}$}\\
    0 & \text{ if $\vr{c}=\vr{y}$}\\
    \vr{c}_p-1 & \text{ if $\vr{c}=\vr{c}_p$}\\
    \vr{c}_{p-1}+1+\vr{c}_0 & \text{ if $\vr{c}=\vr{c}_{p-1}$}\\
    \vr{c} & \text{otherwise}
  \end{cases}
  \right)
  $$
  In particular, by executing $\textbf{updateB}_{d,1}$ following Lemma~\ref{lem:updateE} (if $p=1$), we get from $\alpha_k$ a configuration $\beta$ satisfying the lemma. 
\end{proof}

Since the test-free program $\textbf{evalF}_{d}$ is defined as $\textbf{evalF}_{d,1}\textbf{ or }\cdots\textbf{ or }\textbf{evalF}_{d,d}$, we deduce from Lemma~\ref{lem:evalF} and Lemma~\ref{lem:evalE} the following two corollaries.

\begin{corollary}\label{cor:evalF}
  Let $\alpha$ be a conform configuration and $\beta$ be any configuration such that $\alpha\xrightarrow{\textbf{evalF}_{d}}\beta$. Then $\beta$ is conform and $\ibad{\beta}\geq \ibad{\alpha}$. In particular, if $\beta$ is good then $\alpha$ is good. In that case, $\alpha$ satisfies $\vr{b}$ divides $\vr{x}$, denoting by $(v,n)$ the pair encoded by $\alpha$ then $v\not=\zero_d$, $\reduce_d(v,n)$ is encoded by $\beta$, and $\beta(\vr{x}_{d+1})=\alpha(\vr{x}_{d+1})$.
\end{corollary}

\begin{corollary}\label{cor:evalE}
  Let $\alpha$ be a good configuration satisfying $\vr{b}$ divides $\vr{x}$ and that encodes a pair $(v,n)$ with $v\not=\zero_d$. There exists a good configuration $\beta$ that encodes $\reduce_d(v,n)$ satisfying $\beta(\vr{x}_{d+1})=\alpha(\vr{x}_{d+1})$, and such that $\alpha\xrightarrow{\textbf{evalF}_{d}}\beta$.
\end{corollary}

\section{Ackermannian Preamplifiers}\label{sec:ack}
In this section we prove that the following checking program is a $K$-preamplifier for $K=2^{F_{d+1}(n)}$ of size $O(d4^n)$. 
\progline{\textbf{Ack}_{d,n}}{
  \State $\inc{\vr{c}_0}^{(n)};\inc{\vr{c}_d}^{(n+1)}$
  \State $\inc{\vr{x}};\inc{\vr{x}_1,\ldots,\vr{x}_{d}}^{(2^{n})};\inc{\vr{x}_{d+1}}^{(2^{2n+1})}$
  \Loop
  \State $\inc{\vr{x}};\inc{\vr{x}_1,\ldots,\vr{x}_{d}}^{(2^{n})};\inc{\vr{x}_{d+1}}^{(2^{2n+1})}$
  \EndLoop
  \Loop
  \State $\textbf{evalF}_d$
  \EndLoop
  \Loop
  \State \inc{\vr{y}};\dec{\vr{x}_1,\ldots,\vr{x}_{d+1}}
  \EndLoop
  \Loop
  \State \dec{\vr{c}_0}
  \EndLoop
  \State \test{\vr{x}_{d+1},\vr{b}',\vr{c}_0',\vr{c}_0,\ldots,\vr{c}_d}
}

\begin{lemma}
  The checking program $\textbf{Ack}_{d,n}$ is a $K$-preamplifier for $K=2^{F_{d+1}(n)}$.
\end{lemma}
\begin{proof}
  Let $\ell\geq 1$ and let $\beta$ be the configuration satisfying $\beta(\vr{x},\vr{y},\vr{b})=(\ell,K\ell,K)$ and $\beta(\vr{c})=0$ for every counter $\vr{c}\not\in\{\vr{x},\vr{y},\vr{b}\}$. Let us prove that $0\xrightarrow{\textbf{Ack}_{d,n}}\beta$.

  To do so, we consider an execution of $\textbf{Ack}_{d,n}$ defined as follows. We execute the first loop $2^{F_{d+1}(n)-2n-1}\ell-1$ times (recall that $F_{d+1}(n)\geq 2n+1$). Just after this loop we get a good configuration $\rho_0$ that encodes $((0,\ldots,0,n),n+1)$ and such that $\rho_0(\vr{x}_{d+1})=K\ell$. Then we iterate $\textbf{evalF}_d$ as many times as possible following Lemma~\ref{lem:evalE}. This way we get a sequence $\rho_0,\ldots,\rho_k$ of good configurations such that $\rho_j(\vr{x}_{d+1})=2^{F_{d+1}(n)}\ell$ and such that $\rho_j$ encodes the pair $(v_j,n_j)$ defined as $(v_j,n_j)=\reduce_d^j((n+1)\unit_{d,d},n)$. Since $F_{v_j}(n_j)$ does not depend on $j$, we deduce that $F_{v_j}(n_j)=F_{v_0}(n_0)=F_{d}^{n+1}(n)=F_{d+1}(n)$.

  Since $\rho_k$ is a good configuration, it follows that $\rho_k$ satisfies $\vr{b}=2^{\vr{c}_0}$ and $\vr{x}_{d+1}=2^{\vr{c}_d+\cdots+\vr{c}_0}\vr{x}$. As $\rho_k$ encodes $(v_j,n_j)$, we deduce that $\rho_k(\vr{b})=2^{n_j}$ and $\rho_k(\vr{x})=2^{F_{d+1}(n)-n_j-\norm{v_j}}\ell$.

  We have proved the following equality:
  $$\rho_k(\frac{\vr{x}}{\vr{b}})=2^{F_{v_k}(n_k)-2n_k-\norm{v_k}}$$
  Lemma~\ref{lem:tech} shows that if $v_k\not=\zero_d$ then $\rho_k$ satisfies $\vr{b}$ divides $\vr{x}$. In particular from the good configuration $\rho_k$ we can execute $\textbf{evalF}_d$ one more time following Lemma~\ref{lem:evalE} and get a contradiction with the maximality of $k$. It follows that $v_k=\zero_d$. From $F_{v_k}(n_k)=F_{d+1}(n)$, we deduce that $n_k=F_{d+1}(n)$. Hence $\rho_k(\vr{b})=K$.

  Since $v_k=\zero_d$ and $\rho_k$ is a good configuration, we deduce that $\rho_k$ satisfies $\vr{x}_{d+1}=\cdots=\vr{x}_1$. It follows that $\rho_k(\vr{x}_i)=K\ell$ for every $i\in\{1,\ldots,d+1\}$. Based on this observation, from $\rho_k$ we can execute the third loop $K\ell$ times. Then we execute the last loop $\rho_k(\vr{c}_0)$ times. This way, we get the configuration $\beta$.

  Finally, let us consider a configuration $\beta$ such that $0\xrightarrow{\textrm{Ack}_{d,n}}\beta$ and let us prove that $\beta$ satisfies $\vr{y}\geq \vr{b}\vr{x}$ and $\beta(\vr{c})=0$ for every counter $\vr{c}\not\in\{\vr{x},\vr{y},\vr{b}\}$. Moreover, if the inequality is an equality then let us prove that $\beta(\vr{b})=K$.
  
  We denote by $\rho_0$ the configuration we obtain after executing the first loop. Let $k$ be the number of times the second loop is executed and let us denote by $\rho_j$ the configuration we obtain after the $j$th execution of the second loop for $j\in\{1,\ldots,k\}$. Lemma~\ref{lem:evalF} shows that $\rho_j$ is conform.

  Assume by contradiction that $\rho_k$ is $i$-bad for $i\in\{2,\ldots,d+1\}$. In that case $\rho_k$ satisfies $\vr{x}_i+\vr{y}>2^{\vr{c}_{i-1}}(\vr{x}_{i-1}+\vr{y})$ and $\vr{x}_j+\vr{y}=2^{\vr{c}_j}(\vr{x}_{j-1}+\vr{y})$ for every $j\in\{i+1,\ldots,d+1\}$. In particular $\rho_k(\vr{x}_{d+1})>\rho_k(\vr{x}_{i-1})$. It follows that whatever the number of time the third loop is executed, the configuration $\beta$ satisfy $\beta(\vr{x}_{d+1})>\beta(\vr{x}_{i-1})$ and the last test command $\test{\vr{x}_{d+1}}$ fails. We get a contradiction. It follows that $\rho_k$ is $1$-bad or good.

  Now, assume that $\rho_k$ is $1$-bad. In that case $\rho_k$ satisfies $\vr{x}'=\vr{y}=0$, $\vr{x}>0$, $\vr{b}+\vr{b}'<2^{\vr{c}_0+\vr{c}_0'}$, $\vr{x}_1=2^{\vr{c}_0+\vr{c}_0'}\vr{x}$ and $\vr{x}_i=2^{\vr{c}_{i-1}}\vr{x}_{i-1}$ for every $i\in\{2,\ldots,d+1\}$. Since the execution from $\rho_k$ to $\beta$ only modify $\vr{y},\vr{x}_1,\ldots,\vr{x}_{d+1},\vr{c}_0$, then $\beta$ and $\rho_k$ coincides on the other counters. Since $\beta$ successfully execute the last test commands, we deduce that $\rho_k$ satisfies $\vr{b}'=\vr{c}_0'=\vr{c}_1=\cdots=\vr{c}_d=0$. It follows that $\rho_k$ satisfies $\vr{x}_{d+1}=\cdots=\vr{x}_1$. Since $\beta(\vr{x}_{d+1})=0$, it means that third loop was executed $\rho_k(\vr{x}_{d+1})$ times. We deduce that $\beta(\vr{x}_i)=0$ for every $i\in\{1,\ldots,d\}$. It follows that $\beta$ satisfies $\vr{y}>\vr{b}\vr{x}$ and $\beta(\vr{c})=0$ for every counter $\vr{c}\not\in\{\vr{x},\vr{y},\vr{b}\}$.

  Finally, assume that $\rho_k$ is good. From Lemma~\ref{lem:evalF} we deduce that $\rho_j$ is good for every $j\in\{1,\ldots,k\}$ and denoting by $(v_j,n_j)$ the pair encoded by $\rho_j$ we have $(v_j,n_j)=\reduce_d^j(v_0,n_0)$. In particular $F_{v_j}(n_j)=F_{d+1}(n)$. Since $\rho_k$ is good it satisfies $\vr{x}'=\vr{y}=\vr{b}'=\vr{c}_0'=0$, $\vr{x}>0$, $\vr{b}=2^{\vr{c}_0}$, $\vr{x}_1=2^{\vr{c}_0}\vr{x}$ and $\vr{x}_i=2^{\vr{c}_{i-1}}\vr{x}_{i-1}$ for every $i\in\{2,\ldots,d+1\}$. Since the execution from $\rho_k$ to $\beta$ only modify $\vr{y},\vr{x}_1,\ldots,\vr{x}_{d+1},\vr{c}_0$, then $\beta$ and $\rho_k$ coincides on the other counters. Since $\beta$ successfully execute the last test commands, we deduce that $\rho_k$ satisfies $\vr{c}_1=\cdots=\vr{c}_d=0$. It follows that $v_k=\zero_d$ and from $F_{v_k}(n_k)=F_{d+1}(n)$ we get $n_k=F_{d+1}(n)$. Since $\rho_k$ satisfies $\vr{b}=2^{\vr{c}_0}$ we deduce that $\rho_k(\vr{b})=K$. Since $\rho_k$ satisfies $\vr{x}_{d+1}=\cdots=\vr{x}_1$ and $\beta(\vr{x}_{d+1})=0$, it means that third loop was executed $\rho_k(\vr{x}_{d+1})$ times. We deduce that $\beta(\vr{x}_i)=0$ for every $i\in\{1,\ldots,d\}$. We have proved that $\beta$ satisfies $\vr{y}=\vr{b}\vr{x}$, $\rho_k(\vr{b})=K$, and $\rho_k(\vr{c})=0$ for every counter $\vr{c}\not\in\{\vr{x},\vr{y},\vr{b}\}$.
  
  Therefore $\textrm{Ack}_{d,n}$ is a $K$-preamplifier.
\end{proof}

\begin{corollary}\label{cor:smallamp}
  We can compute in time $O(n d4^n)$ a $K$-preamplifier of size $O(n d4^n)$ with $K=2^{F_{d+1}(n)}$ using $2d+6$ counters such that $d$ of them are safe.
\end{corollary}
\begin{proof}
  The checking program $\textrm{Ack}_{d,n}$ is a $K$-preamplifier of size $O(d4^n)$ computable in time $O(d4^d)$ using $2d+8$ counters such that $\vr{x}',\vr{x}_1,\ldots,\vr{x}_d$ are safe counters. Notice that the counter $\vr{c}_d$ is first initialized to $n$ and then during an execution it is only decremented. It follows that its value can be encoded in the control structure of the program, i.e. by unfolding the program $n$ times. Moreover, notice that during the execution, the value of $\vr{x}_{d}+\vr{y}$ can only increase. It follows that the counter $\vr{x}_d$ can be removed as well by observing that its value (for good configuration) is in fact equals to $\vr{x}_{d+1}2^{-\vr{c}_d}$ and the value of $\vr{c}_d$ is hard coded in the control structure. This way, we get a $K$-preamplifier satisfying the lemma. 
\end{proof}

\section{Complexity Classes Beyond Elementary}\label{sec:complexity}
The reachability problem for general programs equipped with the bounded semantics provides a way to define complexity classes beyond Elementary. In fact, following \cite{Schmitz16toct}, the problem that asks, given a $2$-dimensional general program $M$ of size $n$, whether there exists $\beta$ such that $0\xrightarrow{M}_{\leq K}\beta$ with $K=2^{F_d(n)}$ is $\mathbb{F}_d$-complete for every $d\geq 3$. We deduce as a direct corollary the following theorem.
\begin{theorem}\label{thm:main:fix}
  The reachability problem for $(2d+4)$-dimensional checking programs is $\mathbb{F}_d$-hard for any $d\geq 3$. In particular the reachability problem for $10$-dimensional checking programs is not Elementary.
\end{theorem}
\begin{proof}
  Given a $2$-dimensional program $M$, let $n=\size{M}$. Corollary~\ref{cor:smallamp} shows that we can compute in time $O(n d4^n)$ a $K$-preamplifier $A$ of size $O(n d4^n)$ with $K=2^{F_{d}(n)}$ using $2(d-1)+6$ counters such that $d-1$ of them are safe. Now recall that the checking program $N$ defined as $A\rhd M$ is computable in time $\size{A}+O(\size{M})$. It follows that $N$ is computable in time $O(n d4^n)$. Moreover, as the number of safe counters of $A$ is larger than or equal to $2$, we deduce that the dimension of $N$ is bounded by $2d+4$. Finally, just observe that for every configuration $\beta$ we have $0\xrightarrow{M}_{\leq K}\beta$ if, and only if, $0\xrightarrow{N}\beta$.
\end{proof}
Theorem~\ref{thm:main} is nearly optimal since in \cite{DBLP:conf/lics/LerouxS19} it is proved that the reachability problem for $d$-dimensional checking programs is in $\mathbb{F}_{d+4}$.

\medskip

Finally, let us recall~\cite{Schmitz16toct} that the problem that asks, given a $2$-dimensional general program $M$ of size $n$, whether there exists a configuration $\beta$ such that $0\xrightarrow{M}_{\leq K}\beta$ with $K=2^{F_\omega(n)}$ is $\mathbb{F}_\omega$-complete. We deduce as a direct corollary the following theorem.
\begin{theorem}\label{thm:main}
  The reachability problem for checking programs is $\mathbb{F}_\omega$-hard.
\end{theorem}
\begin{proof}
  Given a $2$-dimensional program $M$, let $n=\size{M}$. Corollary~\ref{cor:smallamp} shows that we can compute in time $O(n^24^n)$ a $K$-preamplifier $A$ of size $O(n^24^n)$ with $K=2^{F_{n+1}(n)}=2^{F_\omega(n)}$. Now recall that the checking program $N$ defined as $A\rhd M$ is computable in time $\size{A}+O(\size{M})$. It follows that $N$ is computable in time $O(n^24^n)$. Finally, just observe that for every configuration $\beta$ we have $0\xrightarrow{M}_{\leq K}\beta$ if, and only if, $0\xrightarrow{N}\beta$.
\end{proof}
Theorem~\ref{thm:main} is optimal since in \cite{DBLP:conf/lics/LerouxS19} it is proved that the reachability problem for checking programs is in $\mathbb{F}_\omega$.

\section{Conclusion}
This paper proves that the reachability problem for checking programs is Ackermannian-complete. It also reduces the gap for the parameterized complexity of the reachability problem in fixed dimension. In order to close this gap, we see two possible research directions:
\begin{itemize}
\item Either we find a primitive recursive algorithm computing from $d,n\in\setN$ a $K$-preamplifier for $K=F_d(n)$ with a dimension $d+O(1)$,
\item Or we find a new algorithm for deciding the reachability problem with a complexity upper bound in $\mathbb{F}_{\frac{d}{2}+O(1)}$.
\end{itemize}

\section*{Acknowledgements}
I gratefully thank Philippe Schnoebelen for helpful comments and encouraging feedback.

\bibliographystyle{plainnat}
\bibliography{biblio}

\begin{thebibliography}{16}
\providecommand{\natexlab}[1]{#1}
\providecommand{\url}[1]{\texttt{#1}}
\expandafter\ifx\csname urlstyle\endcsname\relax
  \providecommand{\doi}[1]{doi: #1}\else
  \providecommand{\doi}{doi: \begingroup \urlstyle{rm}\Url}\fi

\bibitem[Cardoza et~al.(1976)Cardoza, Lipton, and
  Meyer]{DBLP:conf/stoc/CardozaLM76}
E.~Cardoza, R.~J. Lipton, and A.~R. Meyer.
\newblock Exponential space complete problems for petri nets and commutative
  semigroups: Preliminary report.
\newblock In \emph{Proceedings of the 8th Annual {ACM} Symposium on Theory of
  Computing, May 3-5, 1976, Hershey, Pennsylvania, {USA}}, pages 50--54. {ACM},
  1976.
\newblock \doi{10.1145/800113.803630}.

\bibitem[Czerwi\'nski et~al.(2019)Czerwi\'nski, Lasota, Lazi\'c, Leroux, and
  Mazowiecki]{DBLP:conf/stoc/CzerwinskiLLLM19}
W.~Czerwi\'nski, S.~Lasota, R.~Lazi\'c, J.~Leroux, and F.~Mazowiecki.
\newblock The reachability problem for petri nets is not elementary.
\newblock In \emph{Proceedings of the 51st Annual {ACM} {SIGACT} Symposium on
  Theory of Computing, {STOC} 2019, Phoenix, AZ, USA, June 23-26, 2019}, pages
  24--33. {ACM}, 2019.
\newblock \doi{10.1145/3313276.3316369}.

\bibitem[Esparza and Nielsen(1994)]{survey-esparza}
J.~Esparza and M.~Nielsen.
\newblock Decidability issues for petri nets - a survey.
\newblock \emph{Bulletin of the European Association for Theoretical Computer
  Science}, 52:\penalty0 245--262, 1994.

\bibitem[Hack(1975)]{hack75}
M.~H.~T. Hack.
\newblock \emph{Decidability questions for Petri nets}.
\newblock PhD thesis, MIT, 1975.
\newblock URL
  \url{http://publications.csail.mit.edu/lcs/pubs/pdf/MIT-LCS-TR-161.pdf}.

\bibitem[Hopcroft and Pansiot(1979)]{Hopcroft79}
J.~E. Hopcroft and J.-J. Pansiot.
\newblock On the reachability problem for 5-dimensional vector addition
  systems.
\newblock \emph{Theoritical Computer Science}, 8:\penalty0 135--159, 1979.

\bibitem[Karp and Miller(1969)]{DBLP:journals/jcss/KarpM69}
R.~M. Karp and R.~E. Miller.
\newblock Parallel program schemata.
\newblock \emph{J. Comput. Syst. Sci.}, 3\penalty0 (2):\penalty0 147--195,
  1969.
\newblock \doi{10.1016/S0022-0000(69)80011-5}.

\bibitem[Kosaraju(1982)]{Kosaraju82}
S.~R. Kosaraju.
\newblock Decidability of reachability in vector addition systems (preliminary
  version).
\newblock In \emph{{STOC}}, pages 267--281. {ACM}, 1982.
\newblock \doi{10.1145/800070.802201}.

\bibitem[Lambert(1992)]{Lambert92}
J.{-}L. Lambert.
\newblock A structure to decide reachability in {P}etri nets.
\newblock \emph{Theor. Comput. Sci.}, 99\penalty0 (1):\penalty0 79--104, 1992.
\newblock \doi{10.1016/0304-3975(92)90173-D}.

\bibitem[Leroux and Schmitz(2015)]{DBLP:conf/lics/LerouxS15}
J.~Leroux and S.~Schmitz.
\newblock Demystifying reachability in vector addition systems.
\newblock In \emph{30th Annual {ACM/IEEE} Symposium on Logic in Computer
  Science, {LICS} 2015, Kyoto, Japan, July 6-10, 2015}, pages 56--67. {IEEE}
  Computer Society, 2015.
\newblock \doi{10.1109/LICS.2015.16}.

\bibitem[Leroux and Schmitz(2019)]{DBLP:conf/lics/LerouxS19}
J.~Leroux and S.~Schmitz.
\newblock Reachability in vector addition systems is primitive-recursive in
  fixed dimension.
\newblock In \emph{34th Annual {ACM/IEEE} Symposium on Logic in Computer
  Science, {LICS} 2019, Vancouver, BC, Canada, June 24-27, 2019}, pages 1--13.
  {IEEE}, 2019.
\newblock \doi{10.1109/LICS.2019.8785796}.

\bibitem[Mayr(1981)]{Mayr81}
E.~W. Mayr.
\newblock An algorithm for the general petri net reachability problem.
\newblock In \emph{Proceedings of the 13th Annual {ACM} Symposium on Theory of
  Computing, May 11-13, 1981, Milwaukee, Wisconsin, {USA}}, pages 238--246.
  {ACM}, 1981.
\newblock \doi{10.1145/800076.802477}.

\bibitem[Mayr(1984)]{Mayr84}
E.~W. Mayr.
\newblock An algorithm for the general {P}etri net reachability problem.
\newblock \emph{{SIAM} J. Comput.}, 13\penalty0 (3):\penalty0 441--460, 1984.
\newblock \doi{10.1137/0213029}.

\bibitem[Minsky(1967)]{minsky1967computation}
Marvin~L. Minsky.
\newblock \emph{Computation: finite and infinite machines}.
\newblock Prentice-Hall, Inc., 1967.
\newblock URL \url{https://dl.acm.org/citation.cfm?id=1095587}.

\bibitem[Sacerdote and Tenney(1977)]{sacerdote77}
G.~S. Sacerdote and R.~L. Tenney.
\newblock The decidability of the reachability problem for vector addition
  systems (preliminary version).
\newblock In \emph{Proceedings of the 9th Annual {ACM} Symposium on Theory of
  Computing, May 4-6, 1977, Boulder, Colorado, {USA}}, pages 61--76. {ACM},
  1977.
\newblock \doi{10.1145/800105.803396}.

\bibitem[Schmitz(2016{\natexlab{a}})]{DBLP:journals/siglog/Schmitz16}
S.~Schmitz.
\newblock The complexity of reachability in vector addition systems.
\newblock \emph{{SIGLOG} News}, 3\penalty0 (1):\penalty0 4--21,
  2016{\natexlab{a}}.
\newblock URL \url{https://dl.acm.org/citation.cfm?id=2893585}.

\bibitem[Schmitz(2016{\natexlab{b}})]{Schmitz16toct}
Sylvain Schmitz.
\newblock Complexity hierarchies beyond elementary.
\newblock \emph{{TOCT}}, 8\penalty0 (1):\penalty0 3:1--3:36,
  2016{\natexlab{b}}.
\newblock URL \url{http://doi.acm.org/10.1145/2858784}.

\end{thebibliography}


\begin{thebibliography}{23}
\providecommand{\natexlab}[1]{#1}
\providecommand{\url}[1]{\texttt{#1}}
\expandafter\ifx\csname urlstyle\endcsname\relax
  \providecommand{\doi}[1]{doi: #1}\else
  \providecommand{\doi}{doi: \begingroup \urlstyle{rm}\Url}\fi

\bibitem[Bonnet(2011)]{DBLP:conf/mfcs/Bonnet11}
R{\'{e}}mi Bonnet.
\newblock The reachability problem for vector addition system with one
  zero-test.
\newblock In Filip Murlak and Piotr Sankowski, editors, \emph{Mathematical
  Foundations of Computer Science 2011 - 36th International Symposium, {MFCS}
  2011, Warsaw, Poland, August 22-26, 2011. Proceedings}, volume 6907 of
  \emph{Lecture Notes in Computer Science}, pages 145--157. Springer, 2011.
\newblock \doi{10.1007/978-3-642-22993-0\_16}.
\newblock URL \url{https://doi.org/10.1007/978-3-642-22993-0\_16}.

\bibitem[Cardoza et~al.(1976)Cardoza, Lipton, and
  Meyer]{DBLP:conf/stoc/CardozaLM76}
E.~Cardoza, R.~J. Lipton, and A.~R. Meyer.
\newblock Exponential space complete problems for petri nets and commutative
  semigroups: Preliminary report.
\newblock In \emph{Proceedings of the 8th Annual {ACM} Symposium on Theory of
  Computing, May 3-5, 1976, Hershey, Pennsylvania, {USA}}, pages 50--54. {ACM},
  1976.
\newblock \doi{10.1145/800113.803630}.

\bibitem[Czerwi\'nski et~al.(2019)Czerwi\'nski, Lasota, Lazi\'c, Leroux, and
  Mazowiecki]{DBLP:conf/stoc/CzerwinskiLLLM19}
W.~Czerwi\'nski, S.~Lasota, R.~Lazi\'c, J.~Leroux, and F.~Mazowiecki.
\newblock The reachability problem for petri nets is not elementary.
\newblock In \emph{Proceedings of the 51st Annual {ACM} {SIGACT} Symposium on
  Theory of Computing, {STOC} 2019, Phoenix, AZ, USA, June 23-26, 2019}, pages
  24--33. {ACM}, 2019.
\newblock \doi{10.1145/3313276.3316369}.

\bibitem[Czerwinski and Orlikowski(2022)]{DBLP:journals/corr/abs-2104-13866}
Wojciech Czerwinski and Lukasz Orlikowski.
\newblock Reachability in vector addition systems is ackermann-complete.
\newblock In \emph{62st {IEEE} Annual Symposium on Foundations of Computer
  Science, {FOCS} 2022, Denver, Colorado, USA, February 7-10 2022}. {IEEE},
  2022.
\newblock To appear.

\bibitem[Esparza and Nielsen(1994)]{survey-esparza}
J.~Esparza and M.~Nielsen.
\newblock Decidability issues for petri nets - a survey.
\newblock \emph{Bulletin of the European Association for Theoretical Computer
  Science}, 52:\penalty0 245--262, 1994.

\bibitem[Hack(1975)]{hack75}
M.~H.~T. Hack.
\newblock \emph{Decidability questions for Petri nets}.
\newblock PhD thesis, MIT, 1975.
\newblock URL
  \url{http://publications.csail.mit.edu/lcs/pubs/pdf/MIT-LCS-TR-161.pdf}.

\bibitem[Hopcroft and Pansiot(1979)]{Hopcroft79}
J.~E. Hopcroft and J.-J. Pansiot.
\newblock On the reachability problem for 5-dimensional vector addition
  systems.
\newblock \emph{Theoritical Computer Science}, 8:\penalty0 135--159, 1979.

\bibitem[Karp and Miller(1969)]{DBLP:journals/jcss/KarpM69}
R.~M. Karp and R.~E. Miller.
\newblock Parallel program schemata.
\newblock \emph{J. Comput. Syst. Sci.}, 3\penalty0 (2):\penalty0 147--195,
  1969.
\newblock \doi{10.1016/S0022-0000(69)80011-5}.

\bibitem[Kosaraju(1982)]{Kosaraju82}
S.~R. Kosaraju.
\newblock Decidability of reachability in vector addition systems (preliminary
  version).
\newblock In \emph{{STOC}}, pages 267--281. {ACM}, 1982.
\newblock \doi{10.1145/800070.802201}.

\bibitem[Lambert(1992)]{Lambert92}
J.{-}L. Lambert.
\newblock A structure to decide reachability in {P}etri nets.
\newblock \emph{Theor. Comput. Sci.}, 99\penalty0 (1):\penalty0 79--104, 1992.
\newblock \doi{10.1016/0304-3975(92)90173-D}.

\bibitem[Lasota(2021)]{lasota2021improved}
Sławomir Lasota.
\newblock Improved ackermannian lower bound for the vass reachability problem,
  2021.

\bibitem[Leroux and Schmitz(2015)]{DBLP:conf/lics/LerouxS15}
J.~Leroux and S.~Schmitz.
\newblock Demystifying reachability in vector addition systems.
\newblock In \emph{30th Annual {ACM/IEEE} Symposium on Logic in Computer
  Science, {LICS} 2015, Kyoto, Japan, July 6-10, 2015}, pages 56--67. {IEEE}
  Computer Society, 2015.
\newblock \doi{10.1109/LICS.2015.16}.

\bibitem[Leroux and Schmitz(2019)]{DBLP:conf/lics/LerouxS19}
J.~Leroux and S.~Schmitz.
\newblock Reachability in vector addition systems is primitive-recursive in
  fixed dimension.
\newblock In \emph{34th Annual {ACM/IEEE} Symposium on Logic in Computer
  Science, {LICS} 2019, Vancouver, BC, Canada, June 24-27, 2019}, pages 1--13.
  {IEEE}, 2019.
\newblock \doi{10.1109/LICS.2019.8785796}.

\bibitem[Leroux(2022)]{leroux:FOCS22}
Jérôme Leroux.
\newblock The reachability problem for petri nets is not primitive recursive.
\newblock In \emph{62st {IEEE} Annual Symposium on Foundations of Computer
  Science, {FOCS} 2022, Denver, Colorado, USA, February 7-10 2022}. {IEEE},
  2022.
\newblock To appear.

\bibitem[Mayr(1981)]{Mayr81}
E.~W. Mayr.
\newblock An algorithm for the general petri net reachability problem.
\newblock In \emph{Proceedings of the 13th Annual {ACM} Symposium on Theory of
  Computing, May 11-13, 1981, Milwaukee, Wisconsin, {USA}}, pages 238--246.
  {ACM}, 1981.
\newblock \doi{10.1145/800076.802477}.

\bibitem[Mayr(1984)]{Mayr84}
E.~W. Mayr.
\newblock An algorithm for the general {P}etri net reachability problem.
\newblock \emph{{SIAM} J. Comput.}, 13\penalty0 (3):\penalty0 441--460, 1984.
\newblock \doi{10.1137/0213029}.

\bibitem[Minsky(1967)]{minsky1967computation}
Marvin~L. Minsky.
\newblock \emph{Computation: finite and infinite machines}.
\newblock Prentice-Hall, Inc., 1967.
\newblock URL \url{https://dl.acm.org/citation.cfm?id=1095587}.

\bibitem[Rackoff(1978)]{DBLP:journals/tcs/Rackoff78}
C.~Rackoff.
\newblock The covering and boundedness problems for vector addition systems.
\newblock \emph{Theor. Comput. Sci.}, 6:\penalty0 223--231, 1978.
\newblock \doi{10.1016/0304-3975(78)90036-1}.

\bibitem[Reinhardt(2008)]{REINHARDT2008239}
Klaus Reinhardt.
\newblock Reachability in petri nets with inhibitor arcs.
\newblock \emph{Electronic Notes in Theoretical Computer Science},
  223:\penalty0 239--264, 2008.
\newblock ISSN 1571-0661.
\newblock \doi{https://doi.org/10.1016/j.entcs.2008.12.042}.
\newblock URL
  \url{https://www.sciencedirect.com/science/article/pii/S1571066108005057}.
\newblock Proceedings of the Second Workshop on Reachability Problems in
  Computational Models (RP 2008).

\bibitem[Sacerdote and Tenney(1977)]{sacerdote77}
G.~S. Sacerdote and R.~L. Tenney.
\newblock The decidability of the reachability problem for vector addition
  systems (preliminary version).
\newblock In \emph{Proceedings of the 9th Annual {ACM} Symposium on Theory of
  Computing, May 4-6, 1977, Boulder, Colorado, {USA}}, pages 61--76. {ACM},
  1977.
\newblock \doi{10.1145/800105.803396}.

\bibitem[Schmitz(2016{\natexlab{a}})]{DBLP:journals/siglog/Schmitz16}
S.~Schmitz.
\newblock The complexity of reachability in vector addition systems.
\newblock \emph{{SIGLOG} News}, 3\penalty0 (1):\penalty0 4--21,
  2016{\natexlab{a}}.
\newblock URL \url{https://dl.acm.org/citation.cfm?id=2893585}.

\bibitem[Schmitz(2016{\natexlab{b}})]{Schmitz16toct}
Sylvain Schmitz.
\newblock Complexity hierarchies beyond elementary.
\newblock \emph{{TOCT}}, 8\penalty0 (1):\penalty0 3:1--3:36,
  2016{\natexlab{b}}.
\newblock URL \url{http://doi.acm.org/10.1145/2858784}.

\bibitem[Schnoebelen(2010)]{DBLP:conf/mfcs/Schnoebelen10}
Philippe Schnoebelen.
\newblock Revisiting ackermann-hardness for lossy counter machines and reset
  petri nets.
\newblock In Petr Hlinen{\'{y}} and Anton{\'{\i}}n Kucera, editors,
  \emph{Mathematical Foundations of Computer Science 2010, 35th International
  Symposium, {MFCS} 2010, Brno, Czech Republic, August 23-27, 2010.
  Proceedings}, volume 6281 of \emph{Lecture Notes in Computer Science}, pages
  616--628. Springer, 2010.
\newblock \doi{10.1007/978-3-642-15155-2\_54}.
\newblock URL \url{https://doi.org/10.1007/978-3-642-15155-2\_54}.

\end{thebibliography}

\end{document}